\documentclass[11pt]{article}
\usepackage{fullpage}

\usepackage{amsmath}
\usepackage{amsfonts}
\usepackage{amsthm}
\usepackage{amssymb}
\usepackage{dsfont}
\usepackage{hyperref}
\usepackage{graphicx}
\usepackage{thm-restate}
\usepackage[ruled,vlined]{algorithm2e}
\usepackage{cleveref}
\usepackage{color}
\usepackage{thm-restate}
\usepackage{url}

\newcommand{\HW}{\textup{HW}}
\newcommand{\poly}{\textup{poly}}
\newcommand{\polylog}{\textup{polylog}}
\newcommand{\Hhigh}{\mathcal{H}_n^{\textup{high}}}
\newcommand{\Hlow}{\mathcal{H}_n^{\textup{low}}}

\newcommand{\trace}{{\rm Tr}}
\newcommand{\modulo}{{\rm mod}}
\newcommand{\Span}[1]{\operatorname{Span}(#1)}

\newcommand{\GSCON}{\textup{GSCON}}

\newcommand{\lin}[1]{\textup{L}\left(#1\right)}

\newcommand{\unitary}[1]{\textup{U}\left(#1\right)}
\newcommand{\herm}[1]{\textup{Herm}\left(#1\right)}

\newcommand{\comment}[1]{}

\newcommand{\norm}[1]{\left\|\,#1\,\right\|}       
\newcommand{\enorm}[1]{\norm{#1}_{\mathrm{2}}}      
\newcommand{\trnorm}[1]{\norm{#1}_{\mathrm {tr}}}  
\newcommand{\snorm}[1]{\norm{#1}_{\mathrm {\infty}}}    

\newcommand{\lmin}{\lambda_{\min}}
\newcommand{\set}[1]{{\left\{#1\right\}}}    

\newcommand{\abs}[1]{\left\lvert #1 \right\rvert}
\newcommand{\A}{\mathcal{A}}
\newcommand{\ayes}{A_{\textup{yes}}} 
\newcommand{\ano}{A_{\textup{no}}} 
\newcommand{\psihist}{\psi_{\textup{hist}}}

\newcommand{\Hkit}{H_{\textup{kit}}}
\newcommand{\Hamp}{H_{\textup{amp}}}
\newcommand{\LL}{T}
\newcommand{\un}[1]{\ket{\widetilde{#1}}}

\newcommand{\complex}{{\mathbb C}}
\newcommand{\reals}{{\mathbb R}}

\newcommand{\nats}{{\mathbb N}}
\newcommand{\s}{{\mathcal{S}}}

\def\ket#1{ | #1 \rangle}
\def\bra#1{{\langle #1 | }}
\newcommand{\ketbra}[2]{\ket{#1}\!\bra{#2}}        
\newcommand{\braket}[2]{\mbox{$\langle #1  | #2 \rangle$}}

\newcommand{\spa}[1]{\mathcal{#1}}

\newcommand{\hin}{H_{\rm in}}
\newcommand{\hprop}{H_{\rm prop}}
\newcommand{\hout}{H_{\rm out}}
\newcommand{\hstab}{H_{\rm stab}}
\newcommand{\brakett}[2]{\mbox{$\langle #1  | #2 \rangle$}}

\newcommand{\klh}{{k-local Hamiltonian }}

\newcommand{\PP}{\textup{P}}
\newcommand{\NP}{\textup{NP}}
\newcommand{\QMA}{\textup{QMA}}
\newcommand{\QMAo}{\QMA_1}
\newcommand{\StoqMA}{\textup{StoqMA}}
\newcommand{\cqs}{\textup{cq-}\Sigma_2}
\newcommand{\QCMA}{\textup{QCMA}}

\newtheorem{theorem}{Theorem}

\newtheorem{lemma}{Lemma}

\newtheorem{cor}[theorem]{Corollary}
\newtheorem{obs}[theorem]{Observation}

\newtheorem{definition}{Definition}
\newtheorem{exmp}{Example}[section]

\begin{document}

\title{Hardness of approximation for ground state problems}
\author{Sevag Gharibian\footnote{Department of Computer Science, and Institute for Photonic Quantum Systems, Paderborn University, Germany. Email: sevag.gharibian@upb.de.} \and Carsten Hecht\footnote{Department of Computer Science, Paderborn University, Germany. Email: checht2@campus.uni-paderborn.de.}}

\date{\today}


\maketitle

\begin{abstract} 
    After nearly two decades of research, the question of a quantum PCP theorem for quantum Constraint Satisfaction Problems (CSPs) remains wide open. As a result, proving QMA-hardness of approximation for ground state energy estimation, analogous to hardness of approximation for MAX-$k$-CSP, has remained elusive. (QMA is Quantum Merlin-Arthur, a quantum generalization of NP with a quantum proof and quantum verifier.) Recently, it was shown [Bittel, Gharibian, Kliesch, CCC 2023] that a natural problem involving variational quantum circuits is QCMA-hard to approximate within ratio $N^{1-\epsilon}$ for any $\epsilon>0$ and $N$ the input size. (Quantum Classical Merlin-Arthur is QMA, but with a classical proof.) Unfortunately, this problem was not related to quantum CSPs, leaving the question of hardness of approximation for quantum CSPs open.
    
    In this work, we show that if instead of focusing on ground state \emph{energies} (analogous to the optimal number of satisfied clauses), one considers computing \emph{properties} of the ground space (analogous to computing properties of the MAX-$k$-CSP solution space), QCMA-hardness of computing ground space properties can be shown. In particular, we show that it is (1) QCMA-complete within ratio $N^{1-\epsilon}$ to approximate the Ground State Connectivity problem (GSCON), and (2) QCMA-hard within the same ratio to estimate the amount of entanglement of a local Hamiltonian's ground state, denoted Ground State Entanglement (GSE). As a bonus, a simplification of our construction yields NP-completeness of approximation for a natural $k$-SAT reconfiguration problem, to be contrasted with the recent PCP-based PSPACE-hardness of approximation results for a different definition of $k$-SAT reconfiguration [Karthik C.S. and Manurangsi, 2023, and Hirahara, Ohsaka, STOC 2024].
\end{abstract}

\section{Introduction}\label{scn:intro}
Boolean constraint satisfaction problems (CSPs) and their quantum generalization, \emph{local Hamiltonian} problems, have enjoyed a close relationship over the last decade, both in terms of relevance and complexity. 
For starters, just as MAX-$k$-SAT is the canonical NP-complete problem for $k\geq 2$, the \emph{$k$-local Hamiltonian problem ($k$-LH)} is the canonical Quantum Merlin-Arthur (QMA)-complete problem for $k\geq 2$~\cite{kitaevClassicalQuantumComputation2002,kempejulia3localHamiltonianQMAcomplete2003,kempeComplexityLocalHamiltonian2006}. 
Likewise, whereas $2$-SAT can be solved in linear time and $3$-SAT is NP-complete, its analogous local Hamiltonian problem, dubbed \emph{Quantum $k$-SAT}~\cite{bravyiEfficientAlgorithmQuantum2006}, is linear-time solvable for $k=2$~\cite{aradLinearTimeAlgorithm2016,beaudrapLinearTimeAlgorithm2016} and $\QMAo$-complete\footnote{$\QMAo$ is QMA with perfect completeness.} for $k\geq 3$~\cite{bravyiEfficientAlgorithmQuantum2006,gossetQuantum3SATQMA1Complete2013}. 
Fancy a quantum generalization of Schaeffer's dichotomy theorem, which says that Boolean constraint satisfaction problems are either in \PP\ or \NP-complete~\cite{schaeferComplexitySatisfiabilityProblems1978}? Quantum CSPs have one of those too, stating that $k$-LH is either in \PP, \NP-complete, \StoqMA-complete, or \QMA-complete~\cite{cubittComplexityClassificationLocal2016}. 
Of course, the stars have not always aligned between classical and quantum CSPs --- for example, while Max-$2$-SAT on a 1D chain is efficiently solvable, $2$-LH on the line remains  QMA-complete (for sufficiently large, but constant, local dimension)~\cite{aharonovPowerQuantumSystems2009,hallgrenLocalHamiltonianProblem2013}, even when all constraints on the chain are identical (i.e. the translationally invariant setting)~\cite{gottesmanQuantumClassicalComplexity2009,bauschComplexityTranslationallyInvariant2017}.
But by and large, life in the world of classical versus quantum CSPs has been arguably\ldots \emph{peachy}.

That is, of course, until one brings up the topic of PCP theorems, or the closely related question of \emph{hardness of approximation}.
An infamous 2006 blog post issued the community a challenge: To establish a potential quantum PCP theorem. The blog \emph{also} stated that the problem was expected to be {hard}, and this has indeed proven true.
For it was not until 2022 that the field celebrated arguably its first major victory against the quantum PCP conjecture with the establishment of the No Low-Energy Trivial States (NLTS) theorem~\cite{freedmanQuantumSystemsNonkhyperfinite2014,anshuNLTSHamiltoniansGood2023}.
For the first time, this gave (explicit) local Hamiltonians whose ground state energy could not be approximated by constant depth quantum circuits.
This, in turn, is important, because if one believes $\NP\neq \QMA$, then local Hamiltonians $H$ arising from any candidate quantum PCP theorem should not have ``good'' NP witnesses, and constant-depth quantum circuits constitute one possible family of NP witnesses\footnote{This follows from a straightforward light-cone argument, since the local terms of $H$ each act only on $k\in O(1)$ qubits.}. 
With this said, it is unfortunately not clear how current NLTS constructions can be used to encode hard computational problems, leaving the question of hardness of approximation for QMA via a quantum PCP theorem open.

\paragraph{Quantum hardness of approximation without a quantum PCP theorem.} In this paper, we thus study the question --- can one nevertheless obtain hardness of approximation for quantum complexity classes \emph{without} a quantum PCP theorem?\footnote{For clarity, in the quantum setting, proving hardness of approximation results is \emph{not} known to be equivalent to obtaining quantum PCP theorems~\cite{aharonovGuestColumnQuantum2013,natarajanStatusQuantumPCP2024}. Thus, these should be viewed in general as separate questions. Moreover, already classically, certain hardness of approximation results were known long before the establishment of the PCP theorem. A textbook example is that the general TRAVELING SALESMAN PROBLEM cannot be approximated to within any constant factor, which can be shown via a simple reduction from HAMILTONIAN CYCLE~\cite{vaziraniApproximationAlgorithms2003}.}  

The answer is known to be \emph{yes}. 
In 1999, Umans showed how to obtain classical hardness of approximation for the second level of the Polynomial Hierarchy, $\Sigma_2^p$, without utilizing a PCP theorem; rather, the construction utilized dispersers~\cite{umansHardnessApproximatingSpl1999}.
By extending this approach to the quantum setting, the first quantum hardness of approximation result for a quantum complexity class was shown: The quantum SUCCINCT SET COVER problem is hard to approximate for a quantum generalization of the second level of the Polynomial Hierarchy, $\cqs$~\cite{gharibianHardnessApproximationQuantum2012}. 
($\cqs$ is a bounded-error quantum generalization of $\Sigma_2^p$, where the existentially quantified proof is classical, and the universally quantified proof is quantum.)
The shortcoming of this was that, unlike in the classical setting, quantum polynomial hierarchies (the plural is not a typo!) have not yet risen to the level of prominence of their classical cousin, PH; indeed, the area is in its infancy~\cite{yamakamiQuantumNPQuantum2002,lockhartQuantumStateIsomorphism2017,gharibianQuantumGeneralizationsPolynomial2018,falorCollapsiblePolynomialHierarchy2023,grewalEntangledQuantumPolynomial2024,agarwalQuantumPolynomialHierarchies2024,agarwalOracleSeparationsQuantumClassical2024}.
As a result, hardness of approximation for a more established class like QMA or one of its many variants (e.g. QCMA, QMA(2), etc\ldots; see~\cite{gharibianGuestColumn72024} for a survey) would be preferable.

Here, we focus on Quantum Classical Merlin-Arthur (QCMA)~\cite{aharonovQuantumNPSurvey2002}, arguably second in prominence behind QMA, and defined as QMA but with a classical proof.
In~\cite{gharibianHardnessApproximationQuantum2012}, it was observed that a simple modification to the $\cqs$-hardness results therein also yields \QCMA-hardness of approximation for an artificial problem, Quantum Monotone Minimum Satisfying Assignment (QMSA, \Cref{def:QMSA}).
Building on this, the first natural hardness of approximation result for QCMA was given~\cite{bittelOptimalDepthVariational2023}, namely for the problem of estimating the optimal depth of a variational quantum circuit (MIN-VQA). 

\paragraph{Local Hamiltonians, ground spaces, and GSCON.} This brings us full circle to the starting theme of this paper --- CSPs. Specifically, MIN-VQA is not related to a quantum CSP. So, can one show QCMA-hardness of approximation for a natural computational problem for quantum CSPs? 

To answer this, recall first that a \emph{$k$-local Hamiltonian} is an $n$-qubit $2^n\times 2^n$ complex Hermitian matrix $H$ with a succinct representation $H=\sum_i H_i$, where analogous to a MAX-$k$-SAT clause acting on $k$ out of n bits, each $H_i$ is a Hermitian matrix or \emph{clause} acting non-trivially on some subset of $k$ out of $n$ qubits.
The problem $k$-LH then asks, given $H$, to estimate the \emph{ground state energy}, i.e. the smallest eigenvalue $\lmin(H)$. 
The corresponding set of optimal quantum assignments then form the \emph{ground space}, i.e. the span of eigenvectors $\ket{\psi}\in\complex^{2^n}$ with eigenvalue $\lmin(H)$.
The quantum CSP formulation of the quantum PCP conjecture~\cite{aharonovGuestColumnQuantum2013} then posits\footnote{For clarity, in this formulation we are assuming that $H$ is rescaled so that $\norm{H}\leq 1$, for $\norm{\cdot}$ the spectral norm.} that it is QMA-hard to decide if for positive semidefinite $H$, $\lmin(H)=0$ or $\lmin\geq c$ for some fixed constant $c>0$, under the promise that one of these cases holds.
As stated above, however, resolving this conjecture remains a difficult challenge.
Here, we show that by instead focusing on the natural problem of computing properties of the \emph{ground space} (i.e. the \emph{space} of optimal solutions), as opposed to the ground state energy (i.e the \emph{value} attained by all optimal solutions), QCMA-hardness of approximation can be achieved.

\paragraph{Results.} We now introduce the two and a ``half'' computational problems we study, and state our results. The first of these, GSCON, also has a natural classical counterpart from the classical study of \emph{reconfiguration problems}~\cite{gopalanConnectivityBooleanSatisfiability2009}, for which we show an analogous NP-hardness of approximation result; this is the ``half'' we refer to above.\\

\noindent \emph{Result 1: Ground State Connectivity (GSCON).} Introduced in~\cite{gharibianGroundStateConnectivity2015}, Ground State Connectivity (\Cref{def:GSCON}) is the problem of deciding if two ground states of $H$ have a small circuit connecting them through the ground space.
The input is a $k$-local Hamiltonian $H$, two ground states $\ket{\psi}$ and $\ket{\phi}$ (specified via quantum circuits), and a natural number $m$.
The output is whether there exists a sequence of at most $m$ $2$-local unitary gates $(U_1,U_2,\ldots U_m)$, such that two properties hold: 
\begin{enumerate}
    \item (the unitary sequence maps $\ket{\psi}$ to $\ket{\phi}$) $U_m\cdots U_2U_1\ket{\psi}\approx\ket{\phi}$, and 
    \item (all intermediate states are ground states) for all $i\in[m]$, $\ket{\psi_i}:=U_i\cdots U_1\ket{\psi}$ is a ground state of $H$. 
\end{enumerate}

Besides the connection to classical reconfiguration problems, GSCON is physically motivated: Given Hamiltonian $H$ and two ground states $\ket{\psi}$ and $\ket{\phi}$, it asks whether there is an ``energy barrier'' between the states? This question is important in the study of \emph{quantum memories} and \emph{stabilizer codes}.
In quantum memories, the aim is to preserve the state $\ket{\psi}$ of a quantum system, often achieved by logically {encoding} the state within the ground space of a (gapped) Hamiltonian $H$. Since $H$ is gapped, it is unlikely for noise to take $\ket{\psi}$ \emph{out} of the ground space. However, this does not \emph{a priori} prevent a noise process from mapping one codeword to another \emph{through} the ground space. Since noise is usually modeled as local operations, a desirable requirement is thus for the ground space to have an energy barrier between codewords, i.e. no short sequence of local operations should map $\ket{\psi}$ to another codeword. This is precisely the principle behind Kitaev's toy chain model~\cite{kitaevUnpairedMajoranaFermions2001} and toric code~\cite{kitaevFaulttolerantQuantumComputation2003}. The application to quantum stabilizer error-correcting codes is similar (i.e. involves encodings into ground spaces of stabilizer Hamiltonians), except here the aim is typically to \emph{detect} corruption once it occurs, rather than prevent it. However, a stabilizer Hamiltonian $H$ which is a YES instance of GSCON is one for which there is a short sequence of local operations mapping one codeword to another \emph{through} the ground space of $H$ --- this means the corruption process is essentially {undetectable}\footnote{In fact, such undetectable codeword corruption is always possible \emph{if} the corruption process is exponentially long~\cite{gharibianQuantumSpaceGround2023}. More formally, for any local Hamiltonian $H$ and initial and final ground states $\ket{\psi}$ and $\ket{\phi}$ there exists an exponentially long sequence of $2$-qubit gates mapping $\ket{\psi}$ to $\ket{\phi}$, so that at any intermediate point in the computation, we remain exponentially close to the ground space.}.


In terms of complexity, originally shown \QCMA-complete for $k=5$~\cite{gharibianGroundStateConnectivity2015}, it was subsequently found that, surprisingly\footnote{This is in contrast to the fact that $k$-LH for commuting Hamiltonians is \emph{not} known to remain QMA-hard, and in fact is in \NP\ for certain cases~\cite{bravyiCommutativeVersionLocal2005,aharonovComplexityCommutingLocal2011,schuchComplexityCommutingHamiltonians2011,aharonovComplexityTwoDimensional2018,iraniCommutingLocalHamiltonian2023}.}, 
\GSCON\ remains \QCMA-complete even on Hamiltonians with pairwise \emph{commuting} terms~\cite{gossetQCMAHardnessGround2017}. 
In the 1D translation invariant setting, GSCON remains hard, being QCMAEXP-complete~\cite{watsonComplexityTranslationallyInvariant2023}.

We now state our first main result, which shows that \GSCON\ is QCMA-hard to \emph{approximate}.
For this, we reformulate \GSCON\ to have two thresholds $m\leq m'$. 
We add the promise that there always exists a unitary sequence of length $\poly(N')$ satisfying conditions (1) and (2) above. In the YES case, this sequence has length at most $m$, whereas in the NO case, any such sequence has length at least $m'$. That such a sequence always exists even in the NO case is important to ensure the approximation ratio $m'/m$ is well-defined, for if no such unitary sequence satisfying (1) and (2) existed\footnote{Existing \GSCON\ QCMA-hardness constructions indeed have no such sequence in the NO case.}, the ratio can trivially be set to $m'/m\approx\infty$ by choosing arbitrarily large $m'$ for the NO case.

\begin{restatable}{theorem}{thmMain}\label{thm:GSCON}
    For all $\varepsilon > 0$ and $k\ge 5$, GSCON with $k$-local Hamiltonians is QCMA-complete for $\frac{m'}{m}\in \Theta(N'^{1-\epsilon})$, for $N'$ the encoding size of the GSCON instance.
\end{restatable}

\noindent In words, the length of the minimal unitary sequence satisfying conditions (1) and (2) is QCMA-hard to approximate, even within large relative error scaling essentially linearly in the input size, $N'$.
Three comments: (1) The local Hamiltonians in \Cref{thm:GSCON} do not obey any particular geometry constraints.
(2) As in the classical work of Umans~\cite{umansHardnessApproximatingSpl1999}, we write all hardness ratios in this paper relative to the \emph{encoding size} of the input of the problem being reduced to. For classical MAX-SAT problems, for example, the key point is that this typically corresponds to the number of clauses in the instance, and \emph{not} the number of Boolean variables. For GSCON, one can similarly think of the encoding size of an instance as scaling with the number of local terms in the given Hamiltonian.
(3) The reason we are careful to make point (2) is because, in general, the number of clauses can be \emph{superlinear} in $n$. Thus, mapping an approximation ratio given relative to $n$ to one relative to $N'$ can yield an asymptotically smaller ratio. 
This is why simple promise gap amplification techniques, such as taking parallel copies of all clauses, typically do not suffice to achieve hardness of approximation ratios such as in \Cref{thm:GSCON}.\\

\noindent \emph{Result 1.5: Classical analogue of GSCON, Boolean reconfiguration.} The origins of GSCON as a computational problem are rooted in the classical study of \emph{reconfiguration problems}, to which we now make a detour.
First studied in~\cite{gopalanConnectivityBooleanSatisfiability2009}, classical reconfiguration problems have since evolved into an entire research area (see, e.g., \cite{nishimuraIntroductionReconfiguration2018} for an introductory survey). 
In the case of $k$-SAT, the reconfiguration problem is specified as follows: Given an input $k$-SAT formula $\phi$ and satisfying assignments $x$ and $y$, does there exist a sequence of bit flips from $x$ to $y$, so that each intermediate string obtained is also satisfying for $\phi$? 
This is the classical analogue of GSCON, with the exception that the input does not specify a bound $m$ on the number of bit flips allowed --- in the worst case, one may require an \emph{exponential} number of bit flips.
As a result, its complexity is more difficult than NP, being PSPACE-complete~\cite{gopalanConnectivityBooleanSatisfiability2009}.
Recently,
the first PCP theorems for reconfiguration problems have been established~\cite{itoComplexityReconfigurationProblems2011,ohsakaGapPreservingReductions2023,s.InapproximabilityReconfigurationProblems2024,hiraharaProbabilisticallyCheckableReconfiguration2024}, allowing for the first PSPACE-hardness of approximation results for many reconfiguration problems~\cite{s.InapproximabilityReconfigurationProblems2024,hiraharaProbabilisticallyCheckableReconfiguration2024}.

By leveraging our proof technique for GSCON from \Cref{thm:GSCON}, we are able to give our own hardness of approximation result for $k$-SAT reconfiguration, albeit with respect to a different parameter than~\cite{s.InapproximabilityReconfigurationProblems2024,hiraharaProbabilisticallyCheckableReconfiguration2024}. We will first state our result, followed by a discussion of how it differs from~\cite{s.InapproximabilityReconfigurationProblems2024, hiraharaProbabilisticallyCheckableReconfiguration2024}. 
For this, we define the approximation problem Boolean Reconfiguration (BR, \Cref{def:BR}), which is the $k$-SAT reconfiguration problem above, but with thresholds $h$ and $h'$ on the number of bit flips allowed in the YES and NO cases (analogous to GSCON), respectively. 

\begin{restatable}{theorem}{thmBool}\label{thm:HardnessOfApproxReconfiguration}
    BR is NP-hard to approximate for $h'/h\in\Theta(N^{1-\varepsilon})$, for any constant $\varepsilon > 0$ and $N$ the size of the input BR instance.
\end{restatable}

\noindent The main differences between \Cref{thm:GSCON} and \cite{s.InapproximabilityReconfigurationProblems2024,hiraharaProbabilisticallyCheckableReconfiguration2024} are now as follows. 
First, our proof does not rely on a PCP as in~\cite{s.InapproximabilityReconfigurationProblems2024,hiraharaProbabilisticallyCheckableReconfiguration2024}, but rather uses Umans' disperser-based hardness gap construction as a starting point.
Second, our hardness of approximation is with respect to the \emph{length} of the reconfiguration sequence. 
In~\cite{s.InapproximabilityReconfigurationProblems2024,hiraharaProbabilisticallyCheckableReconfiguration2024}, in contrast, there is no length parameter $m$. 
Rather, hardness of approximation therein is shown in the following sense: There exists a constant $\epsilon>0$ such that, in the YES case, there exists a satisfying reconfiguration sequence from $x$ to $y$, and in the NO case, for any reconfiguration sequence, there exists an intermediate string which satisfies at most a $(1-\epsilon)$-fraction of the clauses of $\phi$, for some fixed constant $\epsilon>0$. 
In this sense, \Cref{thm:HardnessOfApproxReconfiguration} is complementary to \cite{s.InapproximabilityReconfigurationProblems2024, hiraharaProbabilisticallyCheckableReconfiguration2024} in both its proof technique and actual result.\\ 

\vspace{-1mm}

\noindent\emph{Result 2. Ground State Entanglement (GSE).} We now return to the quantum setting, and define the 
 second natural quantum problem we study, which asks whether a given local Hamiltonian has a ground state of low entanglement. 
 More formally, we define the Ground State Entanglement problem (GSE, \Cref{def:GSE}), 
for which the input is a local Hamiltonian $H$, prespecified cut $A$ versus $A'$ of the qubits $H$ acts on, and inverse polynomially separated thresholds $\eta_4>\eta_3$.
The output is to decide whether $H$ has a ground state with entanglement entropy at most $\eta_3$ across the $A$ versus $A'$ cut, or whether all ground states have entanglement entropy at least $\eta_4$ across this cut? 
Here, the \emph{entanglement entropy} is a standard entanglement measure, defined for a bipartite pure state $\ket{\psi}_{AA'}$ as $S(\rho_A)$, for reduced state $\rho_A:=\trace_{A'}(\rho_{AA'})$ and $S(\rho):=-\trace(\rho\log \rho$) the von Neumann entropy.

Our second main result is as follows.

\begin{restatable}{theorem}{thmGSE}\label{thm:GSE}
    For all $\varepsilon > 0$ and $k\ge 5$, GSE is QCMA-hard for $\frac{\eta_4}{\eta_3}\in \Theta(N'^{1-\epsilon})$, for $N'$ the encoding size of the GSE instance.
\end{restatable}

\noindent In words, it is QCMA-hard to estimate the minimal entanglement entropy over all states in the ground space, even within large relative error $N'$.
Two comments:
First, unlike \Cref{thm:GSCON}, it is not clear that GSE is also \emph{in} \QCMA. 
This is because verifying the entropy of a given quantum state $\ket{\psi}$ is, in general, at least as hard as Quantum Statistical Zero Knowledge (QSZK)~\cite{watrousLimitsPowerQuantum2002}.
In particular, when one has access to a poly-size quantum circuit preparing $\ket{\psi}$, the problem is QSZK-complete~\cite{ben-aroyaQuantumExpandersMotivation2008}. 
However, in GSE one does \emph{not} have access to such a circuit, since $\ket{\psi}$ is an unknown ground state of $H$! 

Second, there are two previous works~\cite{gheorghiuEstimatingEntropyShallow2024,boulandPublicKeyPseudoentanglementHardness2024} we are aware of which study questions similar to GSE, whose relationship to \Cref{thm:GSE} we now clarify.
The first~\cite{gheorghiuEstimatingEntropyShallow2024} defines the Hamiltonian Quantum Entropy Difference (HQED) problem, which takes two local Hamiltonians as input and a pre-specified cut $A$ versus $A'$, and asks: Among all ground states $\ket{\psi_1}$ of $H_1$ and $\ket{\psi_2}$ of $H_2$ of minimal entanglement entropy, which of $\ket{\psi_1}$ or $\ket{\psi_2}$ has larger entanglement entropy across $A$ versus $A'$? 
Here, the promise gap in entropy is an additive constant (i.e. it is not a relative-error hardness of approximation result), and the hardness obtained is for the Learning With Errors (LWE) problem~\cite{regevLatticesLearningErrors2009}, not QCMA. 
In words, Reference~\cite{gheorghiuEstimatingEntropyShallow2024} shows that if one could efficiently estimate the entanglement entropy difference between ground states of two given Hamiltonians within constant additive error, then one could solve LWE. 

The second relevant work~\cite{boulandPublicKeyPseudoentanglementHardness2024} introduces the Learning Ground State Entanglement Structure (LGSES) problem. 
The question is then, roughly, to determine whether the ground states of a geometrically constrained $H$ have volume law or area law entanglement across a given set of cuts.
First, the differences: Like~\cite{gheorghiuEstimatingEntropyShallow2024}, the hardness results obtained in \cite{boulandPublicKeyPseudoentanglementHardness2024} are for LWE, not QCMA.  
Second, the {input model} to LGSES is not the standard one used in defining complexity classes, but rather has a cryptographic flavor.
Namely, there are two families of computationally indistinguishable Hamiltonians, $\Hlow$ and $\Hhigh$ (which we call ``YES'' and ``NO'' cases for reference), so that when $H$ is drawn from either $\Hlow$ or $\Hhigh$ according to an appropriate random distribution, with overwhelming probability the entanglement entropy will be low or high, respectively. In contrast, our input to GSE is a single Hamiltonian $H$ along with a cut $A$ versus $A'$ (i.e. the standard QCMA input model).
Now, the similarity:
The promise gap for~\cite{boulandPublicKeyPseudoentanglementHardness2024} on the entanglement entropy between the ``YES'' and ``NO'' cases is $O(\polylog(n))$ versus $\Omega(n)$, i.e. a ratio of $(n/\polylog(n))$. 
However, this gap is relative to the number of qubits $n$, not the encoding size of the input, $N'$, as in \Cref{thm:GSE}.
Our understanding\footnote{Due to the multiplication of an $n\times n$ matrix $A$ with $n$-bit string $x$ in Definition 3.2 of the arxiv version\cite{boulandPublicKeyPseudoentanglementHardness2024}.} is that $N'\in\Omega(n^2)$ in \cite{boulandPublicKeyPseudoentanglementHardness2024}, yielding ratio at most $\sqrt{N'}/\polylog(N')$.
Thus, relative to the input model therein, this may be viewed as an LWE-hardness of approximation result, albeit with a quadratically weaker ratio than \Cref{thm:GSE}.

Finally, we remark there are minor overlaps between our techniques and those of \cite{gheorghiuEstimatingEntropyShallow2024,boulandPublicKeyPseudoentanglementHardness2024}, such as the use of circuit-to-Hamiltonian constructions. 
However, the key underlying approach for obtaining our QCMA-hardness results is completely different.


\paragraph{Techniques.} We begin by discussing \Cref{thm:GSCON} for GSCON, followed by \Cref{thm:GSE} for GSE. The proof of \Cref{thm:HardnessOfApproxReconfiguration} for BR is a classical analogue of the proof for GSCON, with some modifications to the clock construction; as such we omit it here.\\

\vspace{-1mm}
\noindent\emph{Techniques for GSCON.} We show a gap-preserving reduction from the artificial QCMA-hard to approximate problem Quantum Monotone Minimum Satisfying Assignment (QMSA, \Cref{def:QMSA}), roughly defined as follows:
The input is a quantum circuit $V=V_{\LL}\cdots V_1$ accepting a non-empty monotone set\footnote{A set $S\subseteq \{0,1\}^n$ is \emph{monotone} if for any $x\in S$, any strings $x'$ obtained by flipping $0$'s to $1$'s in $x$ is also in $S$.} $S\subseteq \{0,1\}^n$, and integer thresholds $0\le g\le g'\le n$. The output is to decide whether there is a low Hamming weight string $x\in \{0,1\}^n$, meaning of Hamming weight at most $g$, accepted by $V$, or if every $x\in \{0,1\}^n$ of Hamming weight at most $g'$ is rejected by $V$, with the promise that one of these two cases holds. 
Our starting point is the QCMA-hardness construction of~\cite{gharibianGroundStateConnectivity2015}, which we must briefly sketch (more details in \Cref{scn:preliminaries}).
Namely, one first applies a \emph{circuit-to-Hamiltonian} construction~\cite{kitaevClassicalQuantumComputation2002} to $V$, obtaining a local Hamiltonian $\Hkit$ encoding the action of $V$ as follows: (1) $\lmin(\Hkit)$ is ``small'' if and only if $V$ accepts the string $x$ it is given, and (2) the quantum state achieving this low energy is the \emph{history state} of $V$, 
\begin{equation}
    \ket{\psihist(x)}:=\frac{1}{\sqrt{\LL+1}}\sum_{t=0}^{\LL} V_t\cdots V_1 \ket{x}_B\otimes\ket{0}_C\otimes\ket{t}_D,
\end{equation}
where $x$ is the string fed to $V$ in register $B$, $C$ is the ancilla, and $D$ is the clock register tracking time, $t$.
To turn this into a GSCON instance, one then appends a $3$-qubit ``GO'' register $E$, and sets the final Hamiltonian to $H=(\Hkit)_{BCD}\otimes P_E$ for $P:=(I-\ketbra{000}{000}-\ketbra{111}{111})_E$.
Finally, the initial and target states are $\ket{\psi}=\ket{0\cdots 0}_{BCD}\ket{000}_E$ and $\ket{\phi}=\ket{0\cdots 0}_{BCD}\ket{111}_E$, respectively.
The intuition is now as follows: In the YES case, an honest prover can first prepare the history state $\ket{\psihist(x)}$ on registers $BCD$ based on a Hamming weight $g$ string $x$, obtaining $\ket{\psihist(x)}_{BCD}\ket{000}_E$, which is still in the null space of $H$. 
The goal is now to flip the GO qubits, but since we are restricted to $2$-local gates, this must be done in two steps, e.g. the first two bits of $E$ are first flipped to $\ket{11}$.
At this point, our state has high overlap with $P_E$, ``switching on'' the Hamiltonian $\Hkit$, which checks that our history state indeed has low energy.
The last bit of $E$ can now be flipped, and the history state uncomputed in order to arrive at target state $\ket{\phi}$.
Soundness follows via the Traversal Lemma (\Cref{l:traversal}), which shows that it is impossible\footnote{Actually, this statement is not entirely true --- it \emph{is} always possible to perform such a mapping without ever having more than inverse exponential overlap with $P$~\cite{gharibianGroundStateConnectivity2015,gharibianQuantumSpaceGround2023}. However, this necessarily requires an \emph{exponential} number of $2$-qubit gates.} to map $\ket{000}_E$ to $\ket{111}_E$ via $2$-local gates without preparing an intermediate state with non-trivial overlap with $P$. And when this overlap occurs, $\Hkit$ will administer a large energy penalty to any history state prepared in $BCD$ based on a string of Hamming weight at most $g'$.

At first glance, this construction already seems to give hardness of approximation when applied to $V$ from QMSA --- in the YES case, the history state is based on a low Hamming weight $g$ string $x$, whereas in the NO case, any low-energy history state must encode a high Hamming weight $g'$ string. Since $g'/g\in \Theta(N^{1-\epsilon})$, this suggests preparing the history state in the YES case requires much fewer gates than in the NO case, which translates into a short unitary sequence for GSCON in the YES case versus a long one in the NO case. 
The catch is that, while it takes $g$ (respectively, $g'$) bits to prepare $x$ in the YES (respectively, NO) case, the cost of preparing the superposition in $\ket{\psihist(x)}$ given $x$ depends on the number of gates $T$ in $V$, which we have no control over! Thus, the approximation ratio achieved scales roughly as $(g+T)/(g'+T)$, which approaches $1$ if $T$ grows superlinearly in the proof size, which is in general the case.

A similar problem was overcome in the QCMA-hardness of approximation result for MIN-VQA by artificially ``amplifying the cost'' of flipping each bit while preparing string $x$~\cite{bittelOptimalDepthVariational2023}, so as to make the cost of preparing $x$ scale roughly as $gT$ rather than $g$.
However, one has much more control for MIN-VQA when designing a reduction, for the following reason. The input therein is a set of local Hamiltonians $\set{G_1,
\ldots G_l}$, and the question is whether given an initial state $\ket{\psi}$, one can prepare a specified target state $\ket{\phi}$ by applying Hamiltonian evolution according to the $G_i$, i.e. $e^{iG_{j_d}\theta_{j_d}}\cdots e^{iG_{j_1}\theta_{j_1}}\ket{\psi}\approx\ket{\phi}$ (here, the $G_i$ are allowed to be applied in any order with repetition, and with arbitrary evolution angles $\theta$)? 
Thus, when designing a hardness of approximation reduction, one can construct the $G_i$ so that flipping a bit of $x$ is artificially costly in terms of the depth $d$ of the sequence of Hamiltonian evolutions. In contrast, for GSCON we do not have the luxury of forcing a prover to evolve according to any particular set of Hamiltonians --- a dishonest prover can simply apply \emph{any} sequence of unitaries desired, so long as each gate is $2$-local.

To overcome this challenge, we begin with the basic GSCON setup outlined above, but move to a \emph{triple} clock construction. 
Briefly, in addition to register $BCDE$ where $B$ stores the $n$-bit string input to $V$, we (1) add two additional clock registers $K$ and $L$ of sizes $4n$ and $2$, respectively, (2) add two registers $F$ and $G$ each of size $n$ and intended to hold copies of proof $x$, and (3) finally an ``amplification'' register $M$.
The basic premise is now that an honest prover proceeds in two phases. 
In the first phase, it manipulates clocks $KL$ to be able to first prepare its desired proof $x$ in $B$. 
The Hamiltonian constraints we design are such that there is a unique time step $t_i$ in which bits $B_i$, $F_i$, and $G_i$ can be acted on. 
At time $t_i$, if desired, the prover flips bits $B_i$, $F_i$ and $G_i$, thus creating triple redundancy.
It is then forced to flip \emph{all} (roughly) $T$ bits in the amplification register $M$ from all zeroes to all ones, because the next time step $t_i+1$ will administer an energy penalty if $B_i$ was set to $1$ but the qubits of $M$ still $0$.
It is now the triple redundancy on $BFG$ that ensures that, once we leave timestep $t_i$, it is impossible to change bit $B_i$, because any single $2$-local gate will break the equality constraints we place between $B_i$, $F_i$, and $G_i$ on all time steps other than $t_i$.
Next, moving to time $t_i+2$, the prover is allowed to uncompute $M$, and subsequently goes to time $t_i+3$, activating another Hamiltonian check that register $M$ is correctly reset to all zeroes.
Finally, we move to time $t_i+4$, and the entire process repeats for the next proof bit $B_{i+1}$.
This first phase continues until we arrive at the final time step on clocks $KL$, which de-activates a constraint preventing the GO register $E$ from being acted upon.
In the second phase, the prover can now build a history state on registers $BCD$, and subsequently flip the GO qubits as in the basic setup to activate $\Hkit$, which checks the history state.
For clarity, it is the flipping of all $T$ bits in the amplification register, $M$, which allows us to achieve our desired approximation ratio.

To make this logic sound, it is imperative for the clocks $KL$ to be carefully designed. The clock sequence we construct for the first phase is
\begin{align}
    &\ket{00000\ldots0}_K\ket{00}_L,\\
    &\ket{10000\ldots0}_K \ket{10}_L,\\
    &\ket{11000\ldots0}_K\ket{11}_L,\\
    &\ket{11100\ldots0}_K\ket{01}_L,\\
    &\ket{11110\ldots0}_K \ket{00}_L\text{, etc}\ldots.
\end{align}
which has three important properties enabling the soundness analysis to work: (1) Moving from a timestep $t_i$ to $t_{i+1}$ requires application of a $2$-qubit unitary on a \emph{unique} pair of qubits $q_{i,1}$ and $q_{i,2}$. (2) It is impossible to jump from $t_i$ to (say) $t_i+2$ via a single $2$-local gate. (3) To satisfy the YES case conditions of GSCON, every intermediate state $\ket{\psi_l}$ computed must have essentially all its amplitude on some \emph{single} timestep $t_{i_l}$ --- for if not, any $2$-qubit unitary attempting to increment the time from $t_i$ to $t_{i+1}$ will necessarily lead to non-trivial weight being placed on an \emph{invalid} clock state, which is penalized.\\

\vspace{-1mm}
\noindent\emph{Techniques for GSE.} We now discuss \Cref{thm:GSE} for estimating ground state entanglement, which again proceeds via a gap-preserving reduction from QMSA. Thus, given a circuit $V$ which either accepts a low Hamming weight $g$ proof or only high Hamming weight $g'$ proofs, our goal is to construct a Hamiltonian with a low-energy state of low entanglement across a pre-specified cut in the YES case, or only highly entangled low-energy states in the NO case. The main idea is to add additional gates to $V$, each of which is controlled on a different proof qubit. Then, for each proof qubit set to $1$, the corresponding added gate creates a Bell pair on a new register. This ensures that when there is a low (respectively, high) Hamming proof in the proof register, $V$ prepares a low- (respectively, high-) entanglement proof across a certain cut. Applying Kitaev's circuit-to-Hamiltonian construction~\cite{kitaevClassicalQuantumComputation2002}  now yields a local Hamiltonian whose history state is correspondingly entangled. This works as desired for the YES case, when the history state is also a low energy state of $\Hkit$. However, the main technical hurdle is that in the NO case, we must show that \emph{all} low energy states are highly entangled, not just the honest prover history state. For example, if the proof register contains a \emph{superposition} over multiple high Hamming weight proofs (each of whose amplitudes can differ, and whose $1$'s can be in different positions), it is not as straightforward to argue that the corresponding ground state generated is highly entangled. 

A natural first idea is to try the standard QCMA trick~\cite{wocjanSeveralNaturalBQPComplete2006} of having $V$ immediately measure its proof register upon reading it to destroy such a superposition, or equivalently immediately copy its proof to a fresh ancilla and apply the principle of deferred measurement. However, this principle requires all qubit copies in the output to be traced out at the end of the computation, which is indeed the case for a general QCMA verifier interested in just measurement statistics on its designated output qubit. In our case, in contrast, we wish to quantify entanglement of the \emph{full} output state across our cut. To bypass this, we instead coherently copy the proof to \emph{two} new registers, and do {not} assume any qubits of our circuit's output are traced out. Crucially, these registers will be on different sides of the cut for the entanglement entropy. Thus, when we consider a Schmidt decomposition for computing the entropy across this cut, it will depend on amplitudes corresponding to standard basis states which contain \emph{both} a classical proof and some standard basis state on the Bell pair register. This allows us to reduce the analysis to the case where the proof register just contains \emph{one} classical state.

A soundness analysis can now be run when the prover sends an arbitrary history state. However, when the prover sends an arbitrary low energy state, two more tricks are required. First, we weight the $\hin+\hprop+\hstab$ terms of $\Hkit$ with a large penalty term and apply the (Extended) Projection Lemma~\cite{kempejulia3localHamiltonianQMAcomplete2003,gharibianComplexitySimulatingLocal2019} to argue that any low-energy state must be close to a history state. We then apply the Fannes inequality \cite{Fannes:1973ddo}, which roughly says that states close in trace distance are also close in entanglement.



\paragraph{Open questions.} We have shown that computing properties of solution spaces to quantum CSPs is QCMA-hard to approximate. As the study of hardness of approximation for quantum complexity classes remains in its infancy, there are many open questions, aside from the natural question of a quantum PCP for QMA.
First, in terms of GSCON, a curious fact is that it remains QCMA-complete even on commuting Hamiltonians~\cite{gossetQCMAHardnessGround2017}. 
Does QCMA-hardness of approximation also hold in this setting?
For GSE, we have studied the estimation of ground state entanglement across a specified cut. 
Can one also show QCMA-hardness of approximation for detecting other entanglement structures, such as area law versus volume law entanglement as in the LWE-hardness results of \cite{boulandPublicKeyPseudoentanglementHardness2024}?
There are two challenges here.
The first is that in contrast to \cite{boulandPublicKeyPseudoentanglementHardness2024}, we must account for the entanglement created by the QMSA verification circuit $V$ embedded in our construction, which we have no control over. This can be partially alleviated by first modifying $V$ to copy its output to an ancilla qubit, subsequently applying $V^\dagger$ to undo the entanglement created by $V$, and then padding by identity gates. However, this is not ideal, as we cannot blow up the size of $V$ by a superlinear factor (otherwise our approximation ratio suffers); thus, we can at best linearly suppress the entanglement generated by $V$. The second (bigger) problem is that even in the YES case for QMSA, the construction of \cite{gharibianHardnessApproximationQuantum2012} can have linear Hamming weight accepting proofs (relative to the size of the proof register, not the instance encoding size), leading to the creation of linearly many Bell pairs in our GSE YES case. In contrast, a 1D area law requires $O(1)$ (or at most $\polylog$) entanglement entropy across cuts. 

Finally, are there other natural QCMA-hard to approximate problems? An excellent candidate is one of the first QCMA-complete problems, approximating the length of the minimum circuit preparing a ground state of a given local Hamiltonian (MIN-CIRCUIT)~\cite{wocjanSeveralNaturalBQPComplete2006}, under the additional promise that a poly-length preparation circuit exists. 
Can one show that this minimum circuit size is QCMA-hard to approximate? 
(Note that unlike the classic Minimum Circuit Size Problem (MCSP), in which one is given a Boolean truth table $T$ as input and asked for a minimum size circuit computing $T$, this question is different in that the input is analogous to a $k$-CSP formula rather than a truth table.)
Our techniques do not obviously extend to MIN-CIRCUIT for a seemingly crucial reason --- for GSCON recall our main challenge was to increase the cost of preparing a history state \emph{without} altering the length of the circuit $V$ on which we apply Kitaev's circuit-to-Hamiltonian construction.
This was critical, because we needed the cost of proof preparation to scale \emph{linearly} with the size of $V$, so that we could translate a hardness gap for QMSA into one for GSCON.
In contrast, MIN-CIRCUIT is a fairly ``bare-bones'' problem, e.g. there is no path through a ground space we can carve out or set of evolution Hamiltonians as in MIN-VQA through which one might attempt to exert control on a dishonest prover. 
Thus, it seems unclear how to amplify the cost of preparing a proof without complicating the circuit $V$ fed into the circuit-to-Hamiltonian construction, which would blow up the encoding size of $\Hkit$ and thus degrade the hardness ratio obtained.

\section{Preliminaries}\label{scn:preliminaries}

\paragraph{Notation.} We use $:=$ to denote a definition. Given $x\in\set{0,1}^n$, $\ket{x}\in(\complex^2)^{\otimes n}$ is the computational basis state labeled by $x$. A vector $\ket{v}$ has Euclidean norm $\enorm{\ket{v}}:=(\sum_i \abs{v_i}^2)^{1/2}$ and  infinity norm $\snorm{\ket{v}}:=\max_{i}\abs{v_i}$. A matrix $A$ has spectral norm $\snorm{A} := \max\{\norm{A\ket{v}}_2 : \norm{\ket{v}}_2 = 1\}$, and trace norm $\trnorm{A}:=\trace{\sqrt{A^\dagger A}}$. For complex Euclidean space $\spa{X}$, let $\lin{\spa{X}}$, $\herm{\spa{X}}$ and $\unitary{\spa{X}}$ denote the sets of linear, Hermitian and unitary operators acting on $\spa{X}$, respectively. $\nats$ is the set of natural numbers, and $[m]:=\set{1,\ldots, m}$. 
For circuit $V=V_{\LL}\cdots V_1$ consisting of $1$- and $2$-qubit gates $V_i$, we define $\abs{V}=\LL$. 

A useful fact relating $\enorm{\cdot}$ to $\trnorm{\cdot}$ for unit vectors $\ket{v}$ and $\ket{w}$ is
\begin{equation}\label{eqn:enorm}
    \trnorm{\ketbra{v}{v}-\ketbra{w}{w}}=2\sqrt{1-\abs{\brakett{v}{w}}^2}\leq 2\enorm{\ket{v}-\ket{w}}.
\end{equation}

\paragraph{Complexity classes.}

\begin{definition}[Quantum-Classical Merlin Arthur (QCMA)~\cite{aharonovQuantumNPSurvey2002}]\label{def:QCMA}
    A promise problem $\A=(\ayes,\ano)$ is in QCMA if there exists a poly-time uniform quantum circuit family $\set{V_n}$ and polynomials $p,q:\nats\rightarrow\nats$ satisfying the following properties. For any input $x\in\set{0,1}^n$, $V_n$ takes in $n+p(n)+q(n)$ qubits as input, consisting of the input $x$ on register $A$, $p(n)$ qubits initialized to a classical proof $\ket{y}\in\set{0,1}^{p(n)}$ on register $B$, and $q(n)$ ancilla qubits initialized to $\ket{0}$ on register $C$. The first qubit of register $C$, denoted $C_1$, is the designated output qubit, a measurement of which in the standard basis after applying $V_n$ yields the following:
    \begin{itemize}
        \item (Completeness) If $x\in \ayes$, $\exists$ proof $y\in\set{0,1}^{p(n)}$ that $V_n$ accepts with probability $\geq 2/3$.
        \item (Soundness) If $x\in \ano$, then $\forall$ proofs $y\in\set{0,1}^{p(n)}$, $V_n$ accepts with probability $\leq 1/3$.
    \end{itemize}
  \end{definition}
\noindent As is standard, throughout this work, we will assume the input $x$ is hard-coded into the circuit for simplicity, and thus concern ourselves with registers $B$, $C$, and $D$.

\paragraph{Circuit-to-Hamiltonian constructions.} We use the $5$-local circuit-to-Hamiltonian construction of~\cite{kitaevClassicalQuantumComputation2002} as a black box. The precise details are not important, only the following properties. The input to the construction is a quantum circuit $V=V_{\LL}\cdots V_1$  (where in this work, each $V_i$ comprising a circuit is at most $2$-local) acting on a proof register $B$ and ancilla  register $C$. 
The output is a $5$-local Hamiltonian $H=\hin+\hout+\hprop+\hstab$ acting on $B\otimes C\otimes D$, for $D$ a clock register (encoded in unary). 
Such constructions are quantum generalizations of the Cook-Levin construction~\cite{cookComplexityTheoremprovingProcedures1971,leonidlevinUniversalSearchProblems1973}, and aim to force low-energy/eigenvalue states of $H$ to look like the quantum analogue of a tableau, the \emph{history state}:
\begin{equation}\label{eqn:hist}
 \ket{\psihist(\phi)}:=\frac{1}{\sqrt{\LL+1}}\sum_{t=0}^{\LL} V_i\cdots V_1 \ket{\phi}_B\otimes\ket{0}_C\otimes\ket{t}_D,
\end{equation}
for $\ket{\phi}$ an arbitrary proof in register $B$. For this, $\hin$ forces the initial state at time $t=0$ to be encoded correctly, $\hout$ penalizes computations which reject at $t=T$, $\hprop$ enforces correct propagation from each timestep $t$ to $t+1$, and $\hstab$ ensures the clock $D$ is correctly encoded. Moreover, $\hin,\hout,\hprop,\hstab\succeq 0$. 

\begin{lemma}[Kitaev~\cite{kitaevClassicalQuantumComputation2002}]\label{l:kitaev}
    The construction of~\cite{kitaevClassicalQuantumComputation2002} maps a quantum circuit $V$ to a $5$-local Hamiltonian $\Hkit$ with parameters $\alpha$ and $\beta$ such that:
    \begin{itemize}
        \item If there exists a proof $\ket{\phi}$ accepted by $V$ with probability at least $1-\epsilon$, then 
            \begin{align} 
                \trace(\Hkit\ket{\psihist(\phi)}\bra{\psihist(\phi)})\leq \alpha:= \frac{\epsilon}{\LL+1}
            \end{align}
        \item If $V$ rejects all proofs $\ket{\psi}$ with probability at least $1-\epsilon$, then the smallest eigenvalue of $\Hkit$ is at least $\beta := \frac{\pi^2(1-\sqrt{\varepsilon})}{2(\LL+1)^3}$.
    \end{itemize}
\end{lemma}
\noindent To ensure $\beta-\alpha\geq 1/\poly$, without loss of generality one first applies standard parallel repetition to $V$ to reduce $\epsilon$ to at most an inverse polynomial.

The ground space of $\Hkit$ is only guaranteed to be a history state when $V$ accepts some proof with certainty; in this case, the history state is the joint null vector of $\hin$, $\hprop$, $\hout$, $\hstab$. In the general case, however, we need to argue that a low energy state must be \emph{close} to a history state. For this, we use an extended version of the Projection Lemma.

\begin{lemma}[Extended Projection Lemma (original in~\cite{kempejulia3localHamiltonianQMAcomplete2003}, extended version in~\cite{gharibianComplexitySimulatingLocal2019})]\label{l:kkr}
	Let $H=H_1+H_2$ be the sum of two Hamiltonians operating on some Hilbert space $\spa{H}=\spa{S}+\spa{S}^\perp$. The Hamiltonian $H_1$ is such that $\spa{S}$ is a zero eigenspace and the eigenvectors in $\spa{S}^\perp$ have eigenvalue at least $J>2\snorm{H_2}$. Let $K:=\snorm{H_2}$. Then, for any $\delta\geq0$ and $\ket{\psi}$ satisfying $\bra{\psi}H\ket{\psi}\leq \lambda(H)+\delta$, there exists a $\ket{\psi'}\in \spa{S}$ such that:
	\begin{itemize}
        \item (Ground state energy bound)
        \[
		\lambda(H_2|_{\spa{S}})-\frac{K^2}{J-2K}\leq \lambda(H)\leq \lambda(H_2|_{\spa{S}}),
	\]
where $\lambda(H_2|_{\spa{S}})$ denotes the smallest eigenvalue of $H_2$ restricted to space $\spa{S}$.
        \item (Ground state deviation bound)
        \[
            \abs{\brakett{\psi}{\psi'}}^2\geq {1-\left(\frac{K+\sqrt{K^2+\delta(J-2K)}}{J-2K}\right)^2}.
            \]
        \item  (Energy obtained by perturbed state against $H$)
        \[
            \bra{\psi'}H\ket{\psi'}\leq\lambda(H)+\delta+2K\frac{K+\sqrt{K^2+\delta(J-2K)}}{J-2K}.
        \]
    \end{itemize}
\end{lemma}

\noindent We will use this in conjunction with the following fact.

\begin{lemma}[Lemma 3 of \cite{gharibianHardnessApproximationQuantum2012}]\label{l:GKgap}
    The smallest non-zero eigenvalue of $\hin+\hprop+\hstab$ is at least $\pi^2/(64T^3)\in\Omega(1/T^3)$,~for $T\geq 1$.
\end{lemma}


\paragraph{Computational problems.} We now define two computational problems we utilize, beginning with GSCON.

\begin{definition}[Ground State Connectivity (GSCON)~\cite{gharibianGroundStateConnectivity2015}]\label{def:GSCON}
    Fix inverse polynomial $\Delta: \nats \rightarrow \reals^+$.
    \begin{itemize}
        \item Input:
        \begin{enumerate}
            \item \klh $H = \sum_i H_i$ acting on n qubits with $H_i \in \herm{(\complex^2)^{\otimes k}}$ satisfying $\snorm{H_i} \le 1$.
            \item Thresholds $\eta_1$, $\eta_2$, $\eta_3$ and $\eta_4 \in \reals$ such that $\eta_2- \eta_1 \ge  \Delta$ and $\eta_4 - \eta_3 \ge \Delta$, and $1^m$, $1^{m'}$ for $m,m' \in \nats$.
            \item Poly-size quantum circuits $U_\psi$ and $U_\phi$ generating ``starting" and ``target" states $\ket\psi$ and $\ket\phi$ (starting from $\ket0^{\otimes n}$), respectively, satisfying $\bra\psi H \ket\psi \le \eta_1$ and $\bra\phi H \ket\phi \le \eta_1$.
            \end{enumerate}
            \item Output:
            \begin{enumerate}
                \item If there exists a sequence of 2-qubit unitaries $(U_i)_{i=1}^m \in \unitary{\complex^2}^{\times m}$ such that:
                \begin{enumerate}
                    \item (Intermediate states remain in low energy space) For all $i \in [m]$ and intermediate states $\ket{\psi_i}:= U_i\cdots U_2 U_1\ket\psi$, one has $\bra{\psi_i} H\ket{\psi_i}\le \eta_1$ and 
                    \item (Final state close to target state) $\norm{U_m\cdots U_1\ket\psi - \ket\phi}_2 \le \eta_3$,
                \end{enumerate}
                then output YES.
                \item If for all sequences of 2-qubit unitaries $(U_i)_{i=1}^{m'} \in \unitary{\complex^2}^{\times m'}$ either:
                \begin{enumerate}
                    \item (An intermediate state has high energy) There exists $i \in [m']$ and an intermediate state $\ket{\psi_i}:= U_i\cdots U_2 U_1\ket\psi$, that has energy $\bra{\psi_i} H\ket{\psi_i}\ge \eta_2$ or 
                    \item (Final state is far from target state) $\norm{U_{m'}\cdots U_1\ket\psi - \ket\phi}_2 \ge \eta_4$,
                \end{enumerate}
                then output NO.
            \end{enumerate}
    \end{itemize}
\end{definition}

As mentioned in \Cref{scn:intro}, GSCON was shown QCMA-hard for $m=m'$~\cite{gharibianGroundStateConnectivity2015}. 
For later use in \Cref{scn:GSCON}, it will be instructive to sketch the proof briefly: 
Take an arbitrary QCMA circuit with proof of size $p(n)$, $q(n)$ many ancilla qubits, and $\LL$ gates. 
Via \Cref{l:kitaev}, we obtain Hamiltonian $\Hkit$ acting on proof, ancilla and clock registers $B$, $C$, and $D$. 
We add a three qubit GO-register $E$ and define starting and target states 
\begin{align}
    \ket\psi =\ket{0}^{p(n)+q(n)+\LL}_{BCD}\ket{0}^{\otimes3}_E,\\[4pt]
    \ket\phi =\ket{0}^{p(n)+q(n)+\LL}_{BCD}\ket{1}^{\otimes3}_E, 
\end{align}
respectively. 
We finally define $H$ as $H:=\Hkit \otimes P$, where $P:=I-\ketbra{000}{000}-\ketbra{111}{111}$ acts on the GO-register $E$. 
Now, to map $\ket{\psi}$ to $\ket{\phi}$, intuitively one must flip the qubits in the GO register using at most $2$-local gates at some point.
This will necessarily create non-trivial support on $P$, so that registers $B$, $C$, and $D$ will be acted on by $\Hkit$, and thus must contain a low energy state. 
In the YES case, this is not a problem, as we can prepare the history state, flip the GO qubits, and uncompute the history state to reach the target state. 
But this is impossible in the NO case, since the ground state energy of the Kitaev Hamiltonian is high. 
To make this rigorous, a crucial lemma that we will also use is the Traversal Lemma \cite{gharibianGroundStateConnectivity2015,gossetQCMAHardnessGround2017}:

\begin{lemma}[Traversal Lemma~\cite{gharibianGroundStateConnectivity2015}]\label{l:traversal}
    Let $S,T \subseteq (\complex^d)^{\otimes n}$ be $k$-orthogonal subspaces. Fix arbitrary states $\ket{v} \in S$ and $\ket{w} \in T$, and consider a sequence of $k$-qubit unitaries $(U_i)=_{i=1}^{m}$ such that
    \begin{align}
        \norm{\ket{w} - U_m \cdots U_1\ket{v}}\le\varepsilon
    \end{align}
    for some $0\le \varepsilon \le 1/2$. Define $\ket{v_i}:=U_i \cdots U_1\ket{v}$ and $P:= I-\Pi_S-\Pi_T$. Then, there exists an $i\in [m]$ such that
    \begin{align}
        \bra{v_i}P\ket{v_i} \ge \left(\frac{1-2\varepsilon}{2m}\right)^2.
    \end{align}
\end{lemma}

Since $\ket{000}$ and $\ket{111}$ are $2$-orthogonal, \Cref{l:traversal} guarantees that at some point, there will be at least $1/\poly$ overlap of an intermediate state and the orthogonal complement of $\Span{\ket{000},\ket{111}}$. 
Thus, as claimed, there must be an intermediate time step with non-trivial support on $P$, at which point the reduced state of $B$, $C$, $D$ is administered a large energy penalty, since we are in the NO case.

The second computational problem we will leverage in our proofs is QMSA. 
For this, we require the following definitions from \cite{gharibianHardnessApproximationQuantum2012}.

\begin{definition}[Monotone set]\label{def:monotone}
    A set $S\subseteq\set{0,1}^n$ is called \emph{monotone} if for any $x\in S$, any string obtained from $x$ by flipping one or more zeroes in $x$ to one is also in $S$.
\end{definition}

\begin{definition}[Quantum circuit accepting monotone set]
    Let $V$ be a quantum circuit consisting of $1$- and $2$-qubit gates, which takes in an $n$-bit classical proof register, $m$-qubit ancilla register initialized to all zeroes, and outputs a single qubit, $q$. For any input $x\in\set{0,1}^n$, we say $V$ \emph{accepts} (respectively, \emph{rejects}) $x$ if measuring $q$ in the standard basis yields $1$ (respectively, $0$) with probability at least $1-\delta$. (If not specified, $\delta=1/3$.) We say $V$ accepts a \emph{monotone set} if the set $S\subseteq\set{0,1}^n$ of all strings accepted by $V$ is monotone (\Cref{def:monotone}).
\end{definition}

\begin{definition}[Quantum Monotone Minimum Satisfying Assignment (QMSA)]\label{def:QMSA}
    Given a quantum circuit $V$ accepting a non-empty monotone set $S\subseteq \{0,1\}^n$, integer thresholds $0\le g\le g'\le n$, output:
    \begin{itemize}
        \item YES if there exists an $x\in \{0,1\}^n$ of Hamming weight at most $g$ accepted  by $V$.
        \item NO if every $x\in \{0,1\}^n$ of Hamming weight at most $g'$ is rejected by $V$.
    \end{itemize}
\end{definition}
\noindent The starting point for our gap-preserving reductions of \Cref{scn:GSCON} and \Cref{scn:GSE} is:
\begin{theorem}[\cite{gharibianHardnessApproximationQuantum2012}]
    For all $\varepsilon > 0$, QMSA is QCMA-hard for $\frac{g'}{g}\in \Theta(N^{1-\varepsilon})$, where $N$ is the input size.
\end{theorem}

\section{Hardness of approximation for GSCON}\label{scn:GSCON}

We now show our first main result, which we restate for convenience:

\thmMain*

\begin{proof} 
    Containment in \QCMA\ is known~\cite{gharibianGroundStateConnectivity2015}. 
    We show \QCMA-hardness of approximation via reduction from QMSA (\Cref{def:QMSA}).
    We divide the proof into four steps: Constructing the GSCON instance from the QMSA instance (\Cref{sscn:GSCONconstruction}), showing completeness (\Cref{sscn:GSCONcompleteness}), showing soundness (\Cref{sscn:GSCONsoundness}), and analyzing the approximation ratio (\Cref{sscn:approxgscon}).

\subsection{Construction}\label{sscn:GSCONconstruction}
    Let $(V,g,g')$ be an instance of QMSA of encoding size $N$ (i.e. requiring $N$ bits to encode), with $V$ taking in an $n$-qubit proof, and let $\LL=\abs{V}$ (the number of $1$- and $2$-qubit gates comprising $V$). 
    Without loss of generality, the completeness and soundness parameters for $V$ are $1-\varepsilon$ and $\varepsilon$ for $\varepsilon \in \Theta(1/2^{N})$~\cite{gharibianHardnessApproximationQuantum2012}. We also assume that the circuit immediately copies the proof to a different register by applying CNOT gates to ancilla qubits, and does not touch the actual proof register afterwards. We first map $(V,g,g')$ to an instance of GSCON in \Cref{ssscn:details}. We then prove properties of the construction in \Cref{ssscn:props} which we will use for completeness and soundness.

    \subsubsection{Details of construction}\label{ssscn:details}
    \noindent\emph{Basic initial setup.}
    We begin with a basic setup similar to~\cite{gharibianGroundStateConnectivity2015}, as outlined in \Cref{scn:preliminaries}.
    Let $\Hkit$ be the Kitaev Hamiltonian (\Cref{l:kitaev}) constructed from $V$, and acting on proof register $B$ of size $n$, ancilla $C$ of size $q(n)$, clock $D$ of size $\LL$.
    For clarity, register $B$ is where the string $x\in\set{0,1}^n$ from \Cref{def:QMSA} is fed.
    Let $P:=I-\ketbra{000}{000}-\ketbra{111}{111}$ be a projector acting on a $3$-qubit GO register $E$.
    Let $W$ be a circuit preparing the history state of the Kitaev Hamiltonian $\Hkit$, given that $B$ contains a string $x$. 
    We note that $W$ can be chosen to have size linear in $N$, the encoding size of $(V,g,g')$. 
    Next  we define the maximum number of $2$-qubit unitaries permitted to prepare the target state in the YES and NO cases, respectively:
\begin{align}
    m &:= 2(2g+g\abs{W}+8n+\abs{W}+1),\\
    m' &:= \frac{1}{2} g' \cdot \abs{W}.
\end{align}

    Recalling from \Cref{l:kitaev} definitions $\alpha := \frac{\varepsilon}{N+1}$ and $\beta := \frac{\pi^2(1-\sqrt{\varepsilon})}{2(N+1)^3}$, we define the step and energy bounds for the GSCON instance:
\begin{align}
\eta_1 := \alpha,\qquad
\eta_2 := \min\left\{\frac{1}{m'^{13}},
\frac{\beta}{6\cdot 64m'^2}\right\},\qquad
\eta_3 := 0,\qquad
\eta_4 := \frac14.
\end{align}
\noindent \emph{Amplifying the cost of preparing each proof bit.} 
To achieve our claimed hardness of approximation result, we henceforth deviate significantly from~\cite{gharibianGroundStateConnectivity2015} in our Hamiltonian construction.
In particular, we need to ``amplify'' the cost of preparing proof bits so as to ``penalize'' proofs of large Hamming weight.
This first requires the addition of five additional registers:
\begin{itemize}
    \item Recalling that $B$ has size $n$, add two registers $F$ and $G$ of size $n$ each. 
    \item Add a second clock register $K$ of size $4n$, and a $2$-qubit register $L$.
    \item Add an amplification register $M$ of size $\abs{W}$. 
\end{itemize}
We define the starting and final states as 
\begin{align}
    \ket\psi &=\ket{0}^{\otimes(n+q(n)+\LL)}_{BCD}\ket{0}^{\otimes3}_E\ket{0}^{\otimes(2n+4n+2+\abs{W})}_{FGKLM} \\[4pt]
    \ket\phi &=\ket{0}^{\otimes(n+q(n)+\LL)}_{BCD}\ket{1}^{\otimes3}_E\ket{0}^{\otimes(2n+4n+2+\abs{W})}_{FGKLM}. 
\end{align}

Finally, we add local Hamiltonian constraints as follows.
In the full construction, we will have two types of constraints: The first type is for checking if a low energy state for the Hamiltonian $\Hkit$ has been prepared (\emph{energy} checks), and the second to force the prover to prepare a state in a manner which requires the application of many local gates for high Hamming weight proofs (\emph{Hamming weight checks}). 
We now discuss both types in detail.\\

\vspace{-1mm}
\noindent \emph{Full set of constraints.} Our full set of Hamiltonian constraints comprising $H$ is:
\begin{enumerate}
    \item (\emph{Energy checks}) The local constraints comprising $\Hkit\otimes P$ check that the string $x$ in proof register $B$, used as the basis for constructing the history state for $\Hkit$, is an accepting proof for the QMSA circuit $V$. 
    We also add a new check so that only when the second clock, $K$, is set to its final timestep, $4n$, can one flip qubits in the GO register $E$:
    \begin{align}
        P_E \otimes \ketbra{0}{0}_{K_{4n}}.\label{eqn:PE}
    \end{align}
    where we use notation $K_{4s}$ to indicate the $4s^{\textup{th}}$ qubit in register $K$.

    \item (\emph{Hamming weight checks}) These constraints ensure that preparing high Hamming weight proofs in $B$ takes a ``large'' number of local unitaries. 
    At a high level, in order to flip any particular proof bit, we force the prover to simulate $4$ phases (corresponding to the four settings of clock register $L$ in \Cref{def:timestep}):
    \begin{enumerate}
        \item (Phase $00$) Depending on the current time step, the prover will be allowed to flip one particular qubit of $B$, along with the copies of this qubit in $F$ and $G$.
        Then, the prover is allowed to flip all qubits in the amplification register, $M$, from $0$ to $1$. 
        (On other timesteps, modifying this particular proof qubit will not be possible due to $3$-local equality constraints we will place between registers $B$, $F$ and $G$.)
        \item (Phase $10$) Check if all qubits in $M$ are either all $0$'s or all $1$'s, and whether they match the prepared proof qubit in $B$.
        \item (Phase $11$) The prover is allowed to flip all qubits in $M$ from $1$ to $0$. 
        \item (Phase $01$) Check if all qubits in $M$ are $0$.
    \end{enumerate}
    Additionally, we define constraints that force the prover to go through all clock states in $K$ in the order we desire, and that force $K$ to be in a correct clock state, meaning the clock qubits will be in a state of the form 
    \begin{align}
            \ket{1}^{\otimes k}\ket{0}^{\otimes 4n-k}.
    \end{align} 
    Formally we can write these constraints as follows, which we partition into ``classes'' (a) through (f)
    :
    \begin{enumerate}
    \item Clock register $K$ has the right form: 
    \begin{equation}
        \ketbra{01}{01}_{K_{i},K_{i+1}} \quad \forall i \in \{1,\ldots,4n-1\}
    \end{equation}
    
    \item Clock register $L$ has the right form, i.e. progresses through the $4$ phases outlined earlier: For all $i \in \set{1,\ldots,n}$,
    \begin{align}
        &\ketbra{0}{0}_{K_1}\otimes (I-\ketbra{00}{00})_L,  \\
        &\ketbra{10}{10}_{K_{4i-3},K_{4i-2}} \otimes (I-\ketbra{10}{10})_L,  \\
        &\ketbra{10}{10}_{K_{4i-2},K_{4i-1}} \otimes (I-\ketbra{11}{11})_L,  \\
        &\ketbra{10}{10}_{K_{4i-1},K_{4i}} \otimes (I-\ketbra{01}{01})_L,  \\
        &\ketbra{10}{10}_{K_{4i},K_{4i+1}} \otimes (I-\ketbra{00}{00})_L \quad\text{for $i<n$},\\
        &\ketbra{1}{1}_{K_{4n}}\otimes (I-\ketbra{00}{00})_L. 
    \end{align}
    
    \item In phases $10$ and $01$, the amplification register $M$ is either all $0$'s or all $1$'s: 
    \begin{align}
        &\ketbra{10}{10}_L \otimes \sum_{i=1}^{|M|-1}(\ketbra{01}{01}+\ketbra{10}{10})_{M_i,M_{i+1}}\\
        &\ketbra{01}{01}_L \otimes \sum_{i=1}^{|M|-1}(\ketbra{01}{01}+\ketbra{10}{10})_{M_i,M_{i+1}}.
    \end{align}
    
    \item In phase $01$, all amplification qubits are $0$: For all $ i \in \set{1,\ldots,n}$,
    \begin{equation}
        \ketbra{10}{10}_{K_{4i-1},K_{4i}} \otimes \ketbra{1}{1}_{M_1}.
    \end{equation}
    
    \item The three proofs cannot be changed outside of exactly one time step  (\Cref{def:timestep}): $\forall i\in \{1,\ldots,n-1\}$,
    \begin{align}
        &(I-\ketbra{000}{000}_{B_1,F_1,G_1}-\ketbra{111}{111}_{B_1,F_1,G_1}) \otimes (I-\ketbra{00}{00}_{K_1,K_2})\label{eqn:proof1}\\
        &(I-\ketbra{000}{000}_{B_{i+1},F_{i+1},G_{i+1}}-\ketbra{111}{111}_{B_{i+1},F_{i+1},G_{i+1}}) \otimes (I-\ketbra{10}{10}_{K_{4i},K_{4i+1}}).\label{eqn:proof2}
    \end{align}

    \item After changing the proof qubits, they should be equal to the amplification qubits, i.e. $ \forall i \in \{0,\ldots,n-1\}$,
    \begin{align}
        &\ketbra{10}{10}_{K_{4i+1},K_{4i+2}} \otimes \ketbra{1}{1}_{B_{i+1}} \otimes \ketbra{0}{0}_{M_1}\\
        &\ketbra{10}{10}_{K_{4i+1},K_{4i+2}} \otimes \ketbra{0}{0}_{B_{i+1}} \otimes \ketbra{1}{1}_{M_1}.
    \end{align}
\end{enumerate}
\end{enumerate}

\noindent \emph{The Final Hamiltonian.} Let $\Hamp$ denote the sum over all local terms corresponding to the energy checks and Hamming weight checks introduced above, except $\Hkit\otimes P$. (Here, ``amp'' stands for amplification.) Recall that $\Hkit=\hin+\hprop+\hout+\hstab$. The final Hamiltonian for our construction is 
\begin{align}\label{eqn:finalconstruction}
    H:=(\mu(\hin+\hprop+\hstab)+\hout)_{BCD} \otimes P_E+(\Hamp)_{BEFGKLM}
\end{align}
for $\mu:=32\pi^4/\beta^2$. The factor \(\mu\) is only used in the projection-lemma argument below.

We note that because of the factor $\mu$, the terms do not have norm 1. But we could instead analyse
\begin{align}
H_{\mathrm{normed}}:= \frac{1}{\mu}H
\end{align}
with energy thresholds 
\begin{align}
\bar\eta_1 := \frac{\eta_1}{\mu} \qquad
\bar\eta_2 := \frac{\eta_2}{\mu}.
\end{align}
For every state \(|\psi\rangle\),
\begin{align}
\langle\psi|H|\psi\rangle \le \eta_i
\quad\Longleftrightarrow\quad
\langle\psi|H_{\mathrm{normed}}|\psi\rangle \le \bar\eta_i.
\end{align}
Therefore, w.l.o.g., we will work with $H$
 going forward.

Also note that all local constraints involved are positive semidefinite, so $H\succeq 0$. We also note that by substituting the Kitaev construction with the 3-local one in \cite{kempejulia3localHamiltonianQMAcomplete2003}, we can get the total locality down to 5 as stated in Theorem \ref{thm:GSCON}. \\

\subsubsection{Properties of the construction}\label{ssscn:props}
The constraints above allow certain operations, depending on how much weight we have on certain configurations of the clock registers. 
In particular, they only allow a strict subset of strings on clock registers $K$ and $L$, which we call \emph{timesteps}: 

\begin{definition}[Timestep $t_i$]\label{def:timestep}
    For $i\in\set{0,\ldots, 4n}$, define unary state
    \begin{align}
        \un{i}:=\ket{1}^{\otimes i}\ket{0}^{\otimes 4n-i}.   
    \end{align} 
    Then, we define timestep $t_i$ as standard basis state $\ket{t_i}:=\un{i}_K\ket{ab}_L$, where
    \begin{align}
        a=
        \begin{cases}
        0, \ \text{for } i = 0,3 \ \modulo \ 4\\
        1, \ \text{for } i = 1,2 \ \modulo  \ 4\\
        \end{cases}
        \quad\text{and}\quad\quad
        b=
        \begin{cases}
        0, \ \text{for } i = 0,1 \ \modulo \ 4\\
        1, \ \text{for } i = 2,3 \ \modulo \ 4\\
        \end{cases}
    \end{align}
\end{definition}

\noindent In words, this clock follows sequence: 
\begin{align}
&\ket{00000\ldots0}_K\ket{00}_L=\un{0}\ket{00}_L,\\
&\ket{10000\ldots0}_K \ket{10}_L=\un{1}\ket{10}_L,\\
&\ket{11000\ldots0}_K\ket{11}_L=\un{2}\ket{11}_L,\\
&\ket{11100\ldots0}_K\ket{01}_L=\un{3}\ket{01}_L,\\
&\ket{11110\ldots0}_K \ket{00}_L=\un{4}\ket{00}_L\text{, etc}\ldots
\end{align}

\begin{obs}[$k$-orthogonality of time steps]\label{obs:orthtime}
    Any pair of timesteps $\ket{t_i}$ and $\ket{t_j}$ for $i<j$ are at least $\abs{j-i+1}$-orthogonal, due to the monotonically increasing Hamming weight in $K$ with increasing $i$.
    Moreover, for any $i$, $\ket{t_i}$ and $\ket{t_{i+1}}$ are (precisely) $2$-orthogonal. 
\end{obs}

\begin{table}[t!]
    \centering
    \begin{tabular}{|c|c|}
        \hline
        \textbf{Variable} & \textbf{Meaning} \\
        \hline
        $N$ & Size of the QMSA instance \\
        \hline
        $N'$ & Size of the GSCON instance \\
        \hline 
        $ \LL$ & Number of gates in the QMSA circuit $V$ \\
        \hline
        $V$ & Circuit of the QMSA instance \\
        \hline
        $W$ & Circuit preparing the history state of V, except for the proof register\\
        \hline
        $g$ & Upper bound for the Hamming weight in the YES case for QMSA\\
        \hline
        $g'$ & Lower bound for the Hamming weight in the NO case for QMSA \\
        \hline
        $m$ & Upper bound for the path size in the YES case for GSCON \\
        \hline
        $m'$ & Lower bound for the path size in the NO case for GSCON \\
        \hline
        $H$ &  Hamiltonian for GSCON \\
        \hline
        $\Hkit$ &  Kitaev Hamiltonian for $V$ \\
        \hline
        $P$ &  Projector acting on the GO register $E$ \\
        \hline
        $\ket{\psi}$ &  Starting state \\
        \hline
        $\ket{\phi}$ &   Final state \\
        \hline
        $U_j$ &   Two qubit unitary \\
        \hline
        $\ket{\psi_j}$ & Intermediate state $U_j\cdots U_1 \ket{\psi}$ \\
        \hline
        $\alpha$ & Upper bound for the ground state energy for $\Hkit$ in the YES Case \\
        \hline
        $\beta$ &  Lower bound for the ground state energy for $\Hkit$ in the NO Case  \\
        \hline
        $\eta_1$  &  Upper bound for the ground state energy for $H$ in the YES Case  \\
        \hline
        $\eta_2$ &   Lower bound for the ground state energy for $H$ in the NO Case  \\
        \hline
        $\eta_3$ &  How close you have to be to the final state in the GSCON YES case \\
        \hline
        $\eta_4$ &  How close you have to be to the final state in the GSCON NO case \\
        \hline
        $B$ & Proof register \\
        \hline
        $C$ & Ancilla register \\
        \hline
        $D$ & Kitaev clock register \\
        \hline
        $E$ & GO register \\
        \hline
        $F$ & Register for the first proof copy \\
        \hline
        $G$ & Register for the second proof copy \\
        \hline
        $K$ & Additional clock register of size $4n$\\
        \hline
        $L$ & Additional clock register of size $2$ \\
        \hline
        $M$ & Amplification register \\
        \hline
        $t_i$ & Timesteps, which are certain standard basis states on registers $K$ and $L$\\
        \hline
    \end{tabular}
    \caption{Overview of variables and their meanings}
    \label{tab:variables}
\end{table}

\noindent\emph{Intuition.} Let us see how an honest prover prepares the proof qubits under these constraints.
By constraint classes (b) and (e), one can prepare the $s$th proof qubit (for $s\in\set{1,\ldots, n}$) precisely when registers $KL$ are in state $\un{4s-4}_K\ket{00}_L$, i.e. the relevant qubits of our state should read
    \begin{align}
        \ket{000}_{B_s,F_s,G_s}\un{4s-4}_K\ket{00}_L\ket{00\ldots0}_M
    \end{align}
    where $\ket{000}_{B_s,F_s,G_s}$ denotes the three copies of proof bit $s$.    
    If the prover wishes to set $B_s=\ket{1}$, it flips the proof and amplification qubits so that it arrives at state
    \begin{align}
        \ket{111}_{B_s,F_s,G_s}\un{4s-4}_K\ket{00}_L\ket{11\ldots1}_M.
    \end{align}
    (If the prover wishes to set $B_s=\ket{0}$, the step above is skipped.)
    Then, the timestep is incremented, activating checks that the amplification qubits are all set to the same bit (class (c)) and that the proof qubit matches the amplification qubits (class (f)):
    \begin{align}\ket{111}_{B_s,F_s,G_s}\un{4s-3}_K\ket{10}_L\ket{11\ldots1}_M\end{align}
    We again increment the timestep, and subsequently flip all amplification qubits back to $0$:
    \begin{align}
        &\ket{111}_{B_s,F_s,G_s}\un{4s-2}_K\ket{11}_L\ket{11\ldots1}_M\quad\rightarrow\\
        &\ket{111}_{B_s,F_s,G_s}\un{4s-2}_K\ket{11}_L\ket{00\ldots0}_M
    \end{align}
    The timestep is incremented again, switching on the check that all amplification qubits equal $0$ (class (c)):
    \begin{align}
        \ket{111}_{B_s,F_s,G_s}\un{4s-1}_K\ket{01}_L\ket{00\ldots0}_M.
    \end{align}
    Finally, the clock is incremented one last time to ``reinitialize'' the state for preparation of the next qubit, qubit $s+1$:
    \begin{align}
        \ket{111}_{B_s,F_s,G_s}\ket{000}_{B_{s+1},F_{s+1},G_{s+1}}\un{4s}_K\ket{00}_L\ket{00\ldots0}_M,
    \end{align}
    where note the contents of $B_sF_sG_s$ remain henceforth unchanged.
    We remind the reader at this point that by \Cref{eqn:PE}, the qubits of the GO register $E$ cannot be flipped until the procedure above reaches time step $\un{4n}_K$.\\

    \vspace{-1mm}
\noindent\emph{One last definition.} In \Cref{sscn:GSCONsoundness}, we will need to argue that intermediate states (in the sense of \Cref{def:GSCON}) have overlap with certain timesteps. We call this overlap \emph{weight}:

\begin{definition}[Weight on a state]
    For a register S and string $x \in \{0,1\}^{\abs{S}}$, we say that a state $\ket\psi$ has \textit{weight} $w_\psi(x)\geq k$ on $\ket x$, if $\ket\psi$ can be written:
    \begin{align}
        \ket\psi = \alpha\ket{x}_{S}\ket{\psi_x} +
        \sum_{\substack{y\in \{0,1\}^{|S|} \\ y\ne x}}( \beta_y\cdot \ket{y}_{S}\ket{\psi_y})
        \qquad
        \text{for}\quad |\alpha|^2 \ge k
    \end{align}
    for some unit vectors $\ket{\psi_z}$ with $z\in\set{0,1}^{\abs{S}}$, and $\abs{\alpha}^2+\sum_y\abs{\beta_y}^2=1$.
\end{definition}

This completes the construction.
For ease of reference, we restate all variables and meanings thereof in \Cref{tab:variables}.


\subsection{Completeness}\label{sscn:GSCONcompleteness}

\begin{algorithm}[t]
    \caption{Preparing the proof qubits}
    \label{alg:my_algorithm}
    
    \mbox{\textbf{Input:} Starting state $\ket\psi= \ket{00\ldots0}$.} \\
    \KwResult{Prepares proof $x^*$ in the proof register, $B$.}
    
    \For{$s=1$ to $n$}{
      
        \If{ $x^*_s = 1$}
        {
        Flip qubits $B_s, F_s, G_s$ and all qubits in $M$ (amplification register)\;
        }
      
      Go to timestep $t_{4s-3}$ by flipping qubit $K_{4s-3}$ and qubits in $L$ according to class (b)\;
    
    Go to timestep $t_{4s-2}$ by flipping qubit $K_{4s-2}$ and qubits in $L$ according to class (b)\;
        
    \If{qubits $B_s, F_s, G_s$ were flipped}{
      Flip all qubits in $M$ back to $0$\;
    }
    Go to timestep $t_{4s-1}$ by flipping qubit $K_{4s-1}$ and qubits in $L$ according to class (b)\;
    Go to timestep $t_{4s}$ by flipping qubit $K_{4s}$ and qubits in $L$ according to class (b)\;
      }
    \end{algorithm}
    
Let us now formalize the intuition of the previous section.
In the YES case, there exists a classical proof $x^*\in\set{0,1}^n$ of Hamming weight $\le g$ that is accepted by $V$. 
When we talk about flipping a qubit, we will mean applying Pauli $X$ to the corresponding qubit. 
We now construct a sequence of $2$-local unitaries mapping $\ket\psi$ to $\ket\phi$ through the ground space of $H$, as per \Cref{def:GSCON}:

\begin{enumerate}
    \item Run \Cref{alg:my_algorithm} to set $B$, $F$, and $G$ each to $x^*$, and $K$ to $\un{4n}$.
    \item Prepare the history state on registers $B,C,D$ via circuit $W$.
    \item Flip the first two GO qubits in register $E$ from $0$ to $1$. This activates the energy checks on the history state. Since we are in the YES case, however, by \Cref{l:kitaev} the history state has energy at most $\alpha=\eta_1$.
    \item Flip the third GO qubit in $E$.
    \item Uncompute the history state on registers $B,C,D$ via circuit $W^\dagger$.
    \item Uncompute values of $B$, $F$, $G$, and $K$ by running \Cref{alg:my_algorithm} in reverse.
\end{enumerate}
By construction, this procedure produces target state $\ket\phi =\ket{0\cdots 0}_{BCD}\ket{111} _E\ket{0\cdots 0}_{FGKLM}$, with all intermediate states of energy at most $\alpha$. (In particular, all steps except 3 yield intermediate states in the null space of $H$.) 
It remains to count the total number of $2$-qubit gates required by this procedure. 
Since $x^\ast$ is of Hamming weight at most $g$, Step 1 uses $2g$ gates for flips in $BFG$, $g\abs{W}$ flips on $M$ for each bit of $B$ it wishes to flip, and $(4n)2$ flips to increment clock registers.
Step 2 takes $\abs{W}$ gates by assumption.
Step 3 and 4 take $2$ gates in total.
Steps 5 and 6 take the same as 1 and 2, respectively.
Thus, in total we require at most $2(2g+g\abs{W}+8n+\abs{W}+1)=m$ gates, as desired.

\subsection{Soundness}\label{sscn:GSCONsoundness}

In the NO case, no classical proof of Hamming weight $\le g'$ is accepted by $V$. 
The proof strategy will be to show that (1) at some point the Kitaev Hamiltonian $\Hkit$ must be ``switched on'', and (2) at this point the proof register will mostly be a superposition of low Hamming weight states. 
The main technical part of the proof will be to show that there exists a subset $S$ of qubits on register $B$ of size $\ge n-g'$, such that measuring any intermediate state (in the sense of \Cref{def:GSCON}) on register $S$ in the standard basis yields all zeros with probability at least $1/2$. 

Intuitively, (1) will follow from the Traversal Lemma (\Cref{l:traversal}), which we will get to later. 
In the meantime, we provide intuition as to why (2) might hold, by first analyzing the simple case where the prover is only allowed to use Pauli $X$ and $I$ gates. (We will subsequently drop this assumption.)
Specifically, we will show that in order to prepare a high Hamming weight proof and subsequently turn off the check of \Cref{eqn:PE} (so that we may alter the contents of the GO register) requires a correspondingly large number of $X$ gates.

\begin{lemma}\label{l:SoundnessGSCONXGates}
    Assume the prover is only allowed to apply gates from set $A:=\set{I, I\otimes X, X\otimes I, X\otimes X}$ to any desired pairs of qubits, as opposed to arbitrary $2$-local unitary gates.  
    Then, in order to prepare a standard basis state on all registers, for which register $B$ contains a proof of Hamming weight at least $g' \in \nats$, and qubit $K_{4n}$ is set to $1$, requires at least $g'\abs{W}$ many gates from $A$.
\end{lemma}
\begin{proof}
    Since all intermediate states must be low energy, if one wants to go from timestep $t_i$ to $t_l$ for $i<l$, then one must also go to all timesteps $t_j$ for all $i<j<l$.
    To flip all three qubits $B_i$,$F_i$, and $G_i$, the clock on registers $K$ and $L$ needs to be in the state $t_{4i-4}$. Say we flipped all three proof qubits. Since we need to be in the last timestep $t_{4n}$ when we turn on the GO register (\Cref{eqn:PE}), we must go through $t_{4i-3}$ and then $t_{4i-1}$. When reaching $t_{4i-3}$, all the amplification qubits in $M$ need to be flipped to $1$. When reaching $t_{4i-1}$, they all need to be flipped back to $0$ again. Flipping all the amplification qubits to $1$ takes at least $\abs{W}/2$ many operations, and since we need to flip them back, we need $2\abs{W}/2=\abs{W}$ flips. Since this happens independently for all proof qubits, and we need to prepare $g'$ proof qubits, we require $g'\abs{W}$ many operations in total. 
\end{proof}

\noindent\emph{The general proof.} To prove \Cref{thm:GSCON}, we need to generalize \Cref{l:SoundnessGSCONXGates} to arbitrary $2$-qubit unitaries instead of just $X$ gates; this is given by \Cref{l:mainlemma} below.
By combining this with the Traversal Lemma (\Cref{l:traversal}), we will be able to complete the proof for soundness. 
In the remainder of this proof, by intermediate states $\ket{\psi_j}$, we mean as defined for \GSCON\ in \Cref{def:GSCON}.

\begin{restatable}[Low Hamming Weight Lemma]{lemma}{mainlemma}\label{l:mainlemma}
    Let $\ket{\psi_j}$ be a state with weight $\ge 1- 1/m'^5$ on timesteps $t_{4n}$ and $t_{4n-1}$. 
    If we allow only $m'$ many operations, all but $g'/2$ many qubits in $B$ have overlap at most $1/m'^3$ with $\ket{1}$. 
    Formally, there exists a subset $S$ of at least $n-g'/2$ many qubits in $B$ such that 
    \begin{equation}
    \forall\; i\in S, \qquad \trace(\ketbra{0}{0}_{B_i}\ketbra{\psi_j}{\psi_j})\ge 1-1/(m')^3.
    \end{equation}
\end{restatable}

\noindent The challenge in proving this is as follows. 
When only using $X$ gates (\Cref{l:SoundnessGSCONXGates}), it was clear that when one is in timestep $t_i$ and wishes to proceed to timestep $t_j$, that one must also visit all time steps $t_k$ for $i<k<j$. 
Unfortunately, it is not \emph{a priori} obvious how to formulate an analogous statement in the arbitrary $2$-qubit unitary case, as the prover could place the state in a superposition of multiple different timesteps. 
However, it turns out something similar \emph{does} hold for the arbitrary unitary case. In \Cref{l:twotimesteps}, we will show that at any point, one can only have large weight on at most two consecutive timesteps. Before giving the statement and proof, we demonstrate the idea via an example.

\begin{exmp}
    Consider state \begin{align}\frac{1}{\sqrt{2}}\ket{000\ldots0}_K\ket{00}_L + \frac{1}{\sqrt{2}}\ket{100\ldots0}_K\ket{10}_L\end{align} which has weight $1/2$ on the timesteps $t_0$ and $t_1$ each. 
    Suppose we wish to put weight $1/4$  onto $\ket{t_2}=\ket{110\ldots0}_K\ket{11}_L$. Since each unitary $U$ is $2$-local, a single such unitary cannot transfer weight from $t_0$ to $t_2$.
    Thus, $U$ must act only on register $K_2L_2$, performing a mapping of form 
    \begin{align}
        \ket{0}_{K_2}\ket{0}_{L_2}\quad\mapsto \quad (\gamma_1 \ket{11} + \gamma_2 \ket{00} + \gamma_3 \ket{10} + \gamma_4 \ket{01})_{K_2,L_2}
    \end{align} 
    with $|\gamma_1|^2 \ge 1/2$.
    But then $U$ also maps $\ket{t_0}=\ket{000\ldots0}_K\ket{00}_L$ to 
    \begin{align}
        \ket{000\ldots0}\ket{00}\mapsto\gamma_1\ket{010\ldots0}\ket{01} + \gamma_2\ket{000\ldots0}\ket{00} + \gamma_3\ket{010\ldots0}\ket{00} + \gamma_4\ket{000\ldots0}\ket{01}.
    \end{align} 
    But this yields weight $\frac{1}{2}\abs{\gamma_1}^2 \ge \frac{1}{4}$ on the illegal timestep $\ket{010\ldots0}_K\ket{01}_L$, taking us out of the low energy subspace (class (a)).
\end{exmp}

We now state and prove \Cref{l:twotimesteps}.

 \begin{lemma}\label{l:twotimesteps}
     If all intermediate states $\ket{\psi_j}$ are low energy, i.e. satisfy $\bra{\psi_j}H\ket{\psi_j}\leq \eta_2$, then there is a constant $c$, such that for all $0\leq j\leq m'$ there exists $i$ with
     \begin{align}
         \ket{\psi_j} = \zeta\ket{t_i}_{KL}\otimes\ket{\phi} + \chi\ket{t_{i+1}}_{KL}\otimes\ket{\phi'} + \sum_{\substack{x\in \set{0,1}^{\abs{K}+\abs{L}} \\ x\ne t_i, t_{i+1}}} \gamma_x \ket{x}_{KL}\otimes\ket{\phi_x}\label{eqn:write}
     \end{align}
such that
\begin{align}
        \sum_{\substack{x\in \{0,1\}^{|K|+|L|} \\ x\ne t_i, t_{i+1}}}  |\gamma_x|^2 < \frac{1}{(m')^c}.
\end{align}
Additionally, if $\ket{\psi_j}$ has this form, and $\ket{\psi_{j+1}}$ has $\Delta\in \reals$ more weight on $t_{i+2}$ than $\ket{\psi_j}$, i.e. $w_{\psi_{j+1}}(t_{i+2})\geq w_{\psi_{j}}(t_{i+2})+\Delta$ then
\begin{align}
        \abs{\zeta}^2 \Delta \in O(\eta_2).
\end{align}
 \end{lemma}

\begin{proof}
    Recall that by definition, $\eta_1$ is exponentially small with respect to $m'$.
    We shall prove by induction over $0 \le j \le m'$ that the following stronger claim holds: 
    Every intermediate state $\ket{\psi_j}$ can be written as in \Cref{eqn:write}
    but with 
    \begin{align}
        \sum_{\substack{x\in \{0,1\}^{|K|+|L|}}} |\gamma_x|^2 < \frac{j}{(m')^{c+1}}.
    \end{align}
    The lemma will then immediately follow.

The base case $j=0$ trivially holds, since (1) our starting state $\ket{\psi_0}=\ket{\psi}$ satisfies $\zeta=1$ and encodes $\ket{t_0}$ in $KL$, and (2) by \Cref{obs:orthtime}, $\ket{t_0}$ and $\ket{t_2}$ are at least $3$-orthogonal, and so the weight on $t_{i+2}$ remains $0$ after applying a single $2$-qubit unitary.
    So, consider the induction step $j \mapsto j+1$:
    
 We first prove the following claim. Let
    \begin{align}
        \Pi_{i} := \ketbra{t_i}{t_i}_{KL} \otimes I_{rest}
    \end{align}
    be the projection onto the $i$-th timestesp. Suppose
\begin{align}
    \|\Pi_{i+2}U_{j+1}\Pi_{i+1}|\psi_j\rangle\|^2 \ge \Delta .
\end{align}
Then
\begin{align}
    w_{\psi_j}(t_i)\Delta \le O(\eta_2).
\end{align}
The analogous statement holds for a transition from \(t_i\) to \(t_{i-1}\). Indeed, by Observation \ref{obs:orthtime}, such a transition can occur only if \(U_{j+1}\) acts on
the two clock qubits on which \(t_{i+1}\) and \(t_{i+2}\) differ, namely
\(K_{i+2}\) and one qubit of \(L\). Let \(a\) and \(b\) be the two-bit strings
on these qubits in \(t_{i+1}\) and \(t_{i+2}\), respectively, and write
\begin{align}
    \lambda := \langle b|U_{j+1}|a\rangle .
\end{align}
Then
\begin{align}
    \|\Pi_{i+2}U_{j+1}\Pi_{i+1}|\psi_j\rangle\|^2
    =
    |\lambda|^2 w_{\psi_j}(t_{i+1}),
\end{align}
and hence \(|\lambda|^2\ge \Delta\). The same local string \(a\) occurs on these two qubits in timestep \(t_i\).
Therefore the same matrix element \(\lambda\) maps the \(t_i\)-component to
the clock string obtained from \(t_i\) by changing these two qubits from
\(a\) to \(b\). This clock string is invalid, since it sets \(K_{i+2}=1\)
while \(K_{i+1}=0\). Thus, up to possible cancellation from the invalid-clock
part of \(|\psi_j\rangle\), whose norm is at most \(\sqrt{\eta_2}\), this
creates invalid-clock norm
\begin{align}
    |\lambda|\sqrt{w_{\psi_j}(t_i)} .
\end{align}
But \(U_{j+1}|\psi_j\rangle=|\psi_{j+1}\rangle\) also has invalid-clock norm at
most \(\sqrt{\eta_2}\). Hence
\begin{align}
    |\lambda|\sqrt{w_{\psi_j}(t_i)}
    \le O(\sqrt{\eta_2}).
\end{align}
Squaring and using \(|\lambda|^2\ge \Delta\), we get
\begin{align}
    w_{\psi_j}(t_i)\Delta \le O(\eta_2),
\end{align}
as claimed. Now suppose \(U\) moves at least
\begin{align}
    \tau := \frac{1}{(m')^{c+1}}
\end{align}
weight out of the current window \(\{t_i,t_{i+1}\}\). Since
\(|\psi_{j+1}\rangle\) has invalid-clock weight at most \(\eta_2\), at least
\(\tau-\eta_2\) of this weight lands on valid clock strings. By
Observation \ref{obs:orthtime}, the only valid timesteps outside the current window that can
receive such weight are \(t_{i-1}\) and \(t_{i+2}\). Hence at least one of
the two transition weights
\begin{align}
    \|\Pi_{i-1}U_{j+1}\Pi_i|\psi_j\rangle\|^2,
    \qquad
    \|\Pi_{i+2}U_{j+1}\Pi_{i+1}|\psi_j\rangle\|^2
\end{align}
is at least
\begin{align}
    \frac{\tau-\eta_2}{2}.
\end{align}
Applying the claim with
\begin{align}
    \Delta = \frac{\tau-\eta_2}{2}
\end{align}
gives
\begin{align}
    w_{\psi_j}(t_i)
    \le
    O\!\left(\frac{\eta_2}{\tau}\right),
\end{align}
assuming the rightward transition occurs; the leftward case is symmetric.
Choosing \(\eta_2\ll \tau^2\), this is at most
\begin{align}
    \frac{1}{10(m')^{c+1}}.
\end{align}
Thus, after shifting the window from \(\{t_i,t_{i+1}\}\) to
\(\{t_{i+1},t_{i+2}\}\), the total weight outside the new window is at most
\begin{align}
    \frac{j}{(m')^{c+1}}
    +
    \frac{1}{10(m')^{c+1}}
    +
    \eta_2
    <
    \frac{j+1}{(m')^{c+1}} ,
\end{align}
for \(m'\) large enough. This completes the induction. Since \(j\le m'\),
the total weight outside two consecutive timesteps is at most
\begin{align}
    \frac{m'}{(m')^{c+1}} = \frac{1}{(m')^c}.
\end{align}

\end{proof}

By setting $\eta_2$ small enough, we can make $c$ an arbitrarily high constant. We note that for our analysis, $c=5$ does suffice.
 A consequence of this is that whenever we need to move weight from one timestep $t_i$ to a later timestep $t_j$, then we will have to have weight almost $1$ on all intermediate timesteps:

\begin{cor}\label{cor:highweightontimestep}
    If an intermediate state $\ket{\psi_k}$ has weight $\ge 1/3m'^4$ on timestep $t_i$ and we want to move weight $\ge 1/3m'^4 - 1/m'^5$ from $t_i$ to $t_j$ for $j>i$, then for all $i<l<j$, there exists intermediate state $\ket{\psi_{k_l}}$ with $k_l > k$ weight at least $1-1/3m'^4$ on timestep  $t_l$. Let $k_l$ be the first integer bigger than $k$ where this happens. Then for $i<l< l' <j$ with corresponding states $\ket{\psi_{k_l}}$ and $\ket{\psi_{k_{l'}}}$, it holds that $k<k_l<k_{l'}$. 
\end{cor}

\begin{proof}
    We need to move weight $1/3m'^4-1/m'^5$ from $t_i$ to $t_{j}$, and every intermediate state can have weight at most $\eta_2$ on illegal timesteps. This means we need to move weight at least 
    \begin{align}
        \frac{1}{3m'^4} - \frac{1}{m'^5} - m'\cdot\eta_2 \ge \frac{1}{4m'^4}
    \end{align} 
    from $t_i$ to $t_{j}$ through legal timesteps. Since we can do $m'$ steps in total, at some point, we will have to move weight $\ge 1/4m'^5$ from $t_{i+1}$ to  $t_{i+2}$. 
    Assume without loss of generality that this happens for the first time with the application of $U_{r+1}$. This implies $\ket{\psi_r}$ has weight at most 
    \begin{align}
        \frac{1}{(m')^c} + m'\cdot \frac{1}{(4m')^5} + m' \cdot \eta_2 
    \end{align}
        on $t_{i+2}$, where $m' \cdot \eta_2$ comes from weight moved through illegal timesteps and $m'\cdot \frac{1}{(4m')^5}$ comes from the fact that this is the first time we are moving weight $\ge 1/4m^{'5}$. By \Cref{l:twotimesteps}, $\ket{\psi_r}$ has at most weight $w$ on $t_i$ for 
    \begin{align}
        w \cdot \frac{1}{4m'^5} \in O(\eta_2).
    \end{align} 
    Since $\ket{\psi_r}$ had weight $\ge 1/m'^5$ on $t_{i+1}$, by \Cref{l:twotimesteps}, all the weight except for $\frac{1}{(m')^c}$ was on $t_{i+1}$ and a neighbouring timestep. 
    Thus, $\ket{\psi_{r}}$ has weight at least
    \begin{align}
        1 - w - \left(\frac{1}{(m')^c} + m'\frac{1}{4m'^5} + m' \cdot \eta_2 \right) - \frac{1}{(m')^c} \ge 1-\frac{1}{3(m')^{4}}
    \end{align} 
    on $t_{i+1}$. The same happens for the other timesteps $t_{l}$ for $i<l<j$.
\end{proof}

Now we can prove the main lemma (\Cref{l:mainlemma}). The main argument will go as follows: We can only prepare a proof qubit in a certain timestep. Since we need to be in the last timestep to turn on the GO register, at some point will have to be in each intermediate timesteps with weight close to $1$. This gives us control over what the amplification register must look like and how often we have to flip all the amplification qubits. This then gives a lower bound for the path length. For convenience, we first restate \Cref{l:mainlemma}.

\mainlemma*
\begin{proof}[Proof of \Cref{l:mainlemma}]
    Let $\ket{\psi}$ be the starting state and $U_1,\dots,U_{k}$ be a sequence of 2-qubit unitaries preparing some state $\ket{\phi'}$. Define the intermediate states $\ket{\psi_i}:= U_i\cdots U_1 \ket{\psi}$. We say that we prepare proof qubit $i$ in step $j$, if after gate $U_j$ we go from  
    
    \begin{align} 
    \ket{\psi_{j-1}} &= \alpha\cdot\ket{000}_{B_i,F_i,G_i}\otimes\ket{\eta_{000}} + \beta\ket{111}_{B_i,F_i,G_i}\otimes\ket{\eta_{111}} +\gamma\cdot\ket{\xi_{j-1}}\qquad\mapsto\\
    \ket{\psi_{j}} &= \alpha'\cdot\ket{000}_{B_i,F_i,G_i}\otimes\ket{\eta_{000}} +
    \beta'\ket{111}_{B_i,F_i,G_i}\otimes\ket{\eta_{111}} +
    \gamma'\cdot\ket{\xi_j},
    \end{align}
     such that  $|\beta'|^2-|\beta|^2 \ge 1/2m'^4$. Say we want to prepare proof qubit $i$ in step $j$. Since the gates are 2-local and $\ket{000}$ and $\ket{111}$ are $2$-orthogonal, that means there must have been weight at least $|\gamma|^2\ge 1/2m'^4$ on states orthogonal to $\ket{000}_{B_i,F_i,G_i}$ and $\ket{111}_{B_i,F_i,G_i}$ in $\ket{\psi_{j-1}}$. Because of the constraints of class (e), we have to be in timestep $t_{4(i-1)}$ with weight $\ge 1/2m'^4-\eta_2$. 
    Without loss of generality, after preparing a proof qubit $i$, we always have weight $\ge 1/2m'^3$ on $\ket{111}_{B_iF_iG_i}$, as otherwise, we would need to prepare proof qubit $i$ again later in order for $\trace(\ketbra{1}{1}_{B_i}\ketbra{\psi_j}{\psi_j})\ge 1/(m')^3$. 
    Because of \Cref{cor:highweightontimestep}, in order to have weight at least $1- \frac{1}{(m'^c)}$ on the last two timesteps, there will have to be weight $\ge 1-1/3m'^4$ first on timestep $t_{4(i-1)+1}$ and then $t_{4(i-1)+3}$. When we have weight $1-1/3m'^4$ on timestep $t_{4(i-1)+1}$ and proof qubit $i$ was already prepared, then by the union bound there is weight at least 
    \begin{align}
        y:=\left(1-1/3m'^4\right)-\left(1-\frac{1}{2m'^3}\right) = \frac{1}{2m'^3}-1/3m'^4
    \end{align}
    on $t_{4(i-1)+1} \otimes \ket{1}_{B_i}$, so also weight $\ge y-\eta_2$ on $t_{4(i-1)+2} \otimes \ket{1}_{B_i}\otimes\ket{1\ldots1}_{M}$. But as soon as we have weight $1-1/3m'^4$ on $t_{4(i-1)+3}$, we also have weight $1-1/3m'^4-\eta_2$ on $\ket{0\ldots 0}_{M}$. So we need to flip all amplification qubits while going from timestep $t_{4(i-1)+1}$ to timestep $t_{4(i-1)+4}$. Since we need to flip the amplification qubits separately for all qubits we prepare, we can prepare at most $g'/2$ many qubits in $m'$ steps (recalling that $m' := \frac{1}{2} g'\abs{W}$ and $\abs{M}=\abs{W}$). 
\end{proof}

Now we can finish the NO case. We are given the starting and final states 
    \begin{align}
        \ket\psi &=\ket{0}^{\otimes(n+q(n)+N)}\ket{0}^{\otimes3}\ket{0}^{\otimes(2n+4n+2+M)}\\
        \ket\phi &=\ket{0}^{\otimes(n+q(n)+N)}\ket{1}^{\otimes3}\ket{0}^{\otimes(2n+4n+2+M)}.\end{align} 
Suppose there is a path $U_1,\ldots,U_{m'}$ such that 
\begin{align}
        \enorm{U_{m'}\cdots U_1\ket\psi-\ket\phi}\le 1/4.
\end{align} 
We apply the Traversal Lemma to $\ket\psi$ and $\ket\phi$ to obtain $j\in[m']$ such that
\begin{equation}
    \bra{\psi_j}P\ket{\psi_j} \ge \left(\frac{1-\frac{1}{2}}{2m'}\right)^2 = \frac{1}{16{m'}^2}.
\end{equation}
By assumption,
\begin{align}
    \eta_2\geq \bra{\psi_j}P\otimes \ketbra{0}{0}_{K_{4n}}\ket{\psi_j}\geq \frac{1}{16m'^2}\trace\left(\frac{P\otimes I\ketbra{\psi_j}{\psi_j}P\otimes I}{\trace(P\otimes I\ketbra{\psi_j}{\psi_j}P\otimes I)}I\otimes \ketbra{0}{0}_{K_{4n}}\right),
\end{align}
from which we conclude $\ket{\psi_j}$ has weight at least $1/(32m'^2)$ on timestep $t_{4n}$. 
By \Cref{l:twotimesteps}, that means $\ket{\psi_j}$ has weight $\ge 1 - 1/m'^5$ on timesteps $4n$ and $4n-1$. Since we have applied at most $m'$ many 2-local unitaries, by \Cref{l:mainlemma}, there exists a set $S$ of size $\ge n-g'/2$ qubits with 
\begin{align}
    \bra{\psi_j}\ketbra{0}{0}_{B_i}\ket{\psi_j} \ge 1-\frac{1}{m'^3}
\end{align} 
for all $i\in S$. By the union bound, in the event measuring all the qubits in the GO register E and in the proof register B, the probability of measuring a state orthogonal to $\ket{000}$ and $\ket{111}$ in the GO register and simultaneously  measuring all 0s on qubits in $S$ greater or equal than 
\begin{align}
    1-\left(1-\frac{1}{16m'^2}\right)-\left(n-\frac{g'}{2}\right)\left(\frac{1}{m'^3}\right) 
    \ge \ \frac{1}{16m'^2} - \frac{1}{32}\frac{m'}{m'^3} 
    \ge \ \frac{1}{32m'^2}.
\end{align}
That means $\ket{\psi_j}$ has weight at least $\frac{1}{6}\cdot \frac{1}{32m'^2}$ on a state of the form
\begin{align}\label{eqn:2}
   \ket{\psi'} := \sum_{\substack{x}} \alpha_{x}\ket{x}_B\ket{z}_E\ket{\psi_{x}'},
\end{align}
where $\ket{z}$ is one of the 6 standard basis states on the GO register $E$ orthogonal to $\ket{000}$ and $\ket{111}$, and $\ket{x}$ is are strings of Hamming weight $\le g'/2$. 
Now we argue why it is sufficient to lower bound the energy of $\ket{\psi'}$ in order to lower bound the energy of $\ket{\psi_j}$. We can write $$\ket{\psi_j}= \sum_{\substack{x,y}} \alpha_{x,y}\ket{x}_B\ket{y}_E\ket{\psi_{x,y}},$$ where x and y go through all the standard basis states. Since the circuit $V$ does not act on the proof register except for CNOT gates in the beginning, we have that 
$${\bra{x}_B\bra{y}_E\bra{\psi_{x,y}}} (H_{kit}\otimes P) \ket{x'}_B\ket{y'}_E\ket{\psi_{x',y'}}=0$$
when $(x,y)\neq (x',y')$. This implies that in order to lower bound the total energy of $\ket{\psi_j}$, it is sufficient to lower bound the energy of the $\alpha_{x,y}\ket{x}_B\ket{y}_E\ket{\psi_{x,y}}$ terms. In our case, we will lower bound the energy of $\sum_{\substack{x'}} \alpha_{x',z}\ket{x'}_B\ket{z}_E\ket{\psi_{x',z}}$, where $z$ is a standard basis state orthogonal to $\ket{000}$ and $\ket{111}$ and the $x'$ will be low Hamming weight strings. This then yields the lower bound for the energy of $\ket{\psi_j}$.

Regarding $\ket{\psi'_{BCD}}$, to conclude the proof, we would ideally desire $\bra{\psi'}(\Hkit)_{BCD}\ket{\psi'}\geq \beta$, since we are in a NO case for QMSA. However, this is not necessarily true, since recall the QMSA circuit $V$ accepts proofs of Hamming weight $\geq g'$ in the NO case, and so $\Hkit\not\succeq \beta I$. To achieve approximately this bound, we thus finally use the parameter $\mu$ (the only part of our proof requiring it), introduced in defining $H:=(\mu(\hin+\hprop+\hstab)+\hout)_{BCD} \otimes P_E+(\Hamp)_{BEFGKLM}$. Namely, we use the Extended Projection Lemma (\Cref{l:kkr}) with $H_1=\mu(\hin+\hprop+\hstab)$ and $H_2=\hout$, where $\spa{S}$ denotes the null space of $\hin+\hprop+\hout$, which recall is precisely the span of all history states. Combining this with \Cref{l:GKgap} and the assumption that $\bra{\psi'}\mu(\hin+\hprop+\hstab)+\hout\ket{\psi'}\leq \beta$ (for otherwise we can skip the use of \Cref{l:kkr} and proceed directly to \Cref{eqn:1}), we conclude that there exists a history state $\ket{\psihist}$ such that 
\begin{align}
    \abs{\braket{\psi'}{\psihist}}^2\geq 1-\frac{2\pi^4}{\mu}.
\end{align}
By \Cref{eqn:enorm}, 
\begin{align}\label{eqn:trnorm}
    \trnorm{\ketbra{\psi'}{\psi'}-\ketbra{\psihist}{\psihist}}\leq \frac{2\sqrt{2}\pi^2}{\sqrt{\mu}}.
\end{align}

Combining these results, we finally have
\begin{align}
   \bra{\psi_j}H\ket{\psi_j}  &\ge  \bra{\psi_j}(\mu(\hin+\hprop+\hstab)+\hout)_{BCD} \otimes P_E\ket{\psi_j} \label{eqn:1}\\
   &\ge \frac{1}{6\cdot 32m'^2} \bra{\psi'}(P\otimes \Hkit)\ket{\psi'}\\
   &\ge \frac{1}{6\cdot 32m'^2} \cdot \left(\beta -\frac{2\sqrt{2}\pi^2}{\sqrt{\mu}}\right)\\
   &\ge \eta_2,
\end{align}
where the first inequality follows since $\Hamp\succeq 0$, the second by \Cref{eqn:2}, and the third by \Cref{eqn:trnorm}, the H\"{o}lder inequality and the fact that $\snorm{\hout}=1$. 
\end{proof}

\subsection{Approximation ratio}\label{sscn:approxgscon}

Fix any $\epsilon>0$. Recall that $N$ is the encoding size of the QMSA instance with ratio $g'/g\in \Theta(N^{1-\epsilon})$, and let $N'$ denote the encoding size of our constructed GSCON instance. We wish to show that $m'/m\in\Theta(N'^{1-\epsilon})$. For this, note first that since $\abs{W}\in\Theta(N)$, 
\begin{align}
    \frac{m'}{m}\in\Theta\left(\frac{\frac{1}{2} g'N}{4g+2g\abs{W}+16n+2N+2}\right)\in\Theta\left(\frac{\frac{1}{2}g'}{2g +\Theta(1)}\right)\in\Theta(N^{1-\epsilon}).
\end{align}
Thus, the claim will follow if $N'\in\Theta(N)$. This, in turn, is true because the total number of constraints (each of constant size) present in $H$ is $O(N)$. (Specifically, the $\Hkit$ terms contribute $O(N)$ terms, and each class from (a)-(f) contributes $\Theta(n)\in O(N)$ terms.)
\section{Hardness of approximation for GSE}\label{scn:GSE}
We next prove \Cref{thm:GSE}, i.e. hardness of approximation for GSE. We begin by formally defining the latter.

\begin{definition}[Ground State Entanglement(GSE)]\label{def:GSE}
    Fix an inverse polynomial $\Delta: \nats \mapsto \reals^+$. Denote the function for the von Neumann entropy as $\s$.
    \begin{itemize}
        \item Input:
        \begin{enumerate}
            \item \klh $H = \sum_i H_i$ acting on n qubits with $H_i \in \herm{(\complex^2)^{\otimes k}}$ satisfying $\norm{H_i} \le 1$.
            \item Thresholds $\eta_1$, $\eta_2$, $\eta_3$ and $\eta_4 \in \reals$ such that $\eta_2- \eta_1 \ge  \Delta$, $\eta_4 - \eta_3 \ge \Delta$ and $0 \le \eta_3 \le \eta_4 \le \min\{\abs{A},\abs{A'}\}$
            \item A partition $(A,A')$ dividing the n qubits into two sets.
            \end{enumerate}
            \item Output:
            \begin{enumerate}
                \item If there exists a state $\ket\psi \in (\complex^2)^{\otimes n}$ with 
                    (low energy) $\bra{\psi} H\ket{\psi}\le \eta_1$ and  
                     (low entanglement) $\s(\ketbra{\psi}{\psi}_A) \le \eta_3$,
                then output YES.
                \item If for all states $\ket\psi \in (\complex^2)^{\otimes n}$, either
                (high energy) $\bra{\psi} H\ket{\psi}\ge \eta_2$ or  
                (high entanglement) $\s(\ketbra{\psi}{\psi}_A) \ge \eta_4$,
                then output NO.
            \end{enumerate}
    \end{itemize}
\end{definition}

We restate \Cref{thm:GSE} for convenience.
\thmGSE*


\noindent For its proof, we will need the Fannes inequality \cite{Fannes:1973ddo}, which converts trace distance bounds into entropy distance bounds.

\begin{lemma}\label{eqn:Fannes}(Fannes inequality~\cite{Fannes:1973ddo})
    Suppose $\rho$ and $\sigma$ are density operators such that the trace distance between them satisfies $\trnorm{\rho-\sigma} \le 1/e$. Then 
    \begin{align}
        \abs{\s(\rho)-\s(\sigma)} \le \trnorm{\rho-\sigma}\log(d) + \eta(\trnorm{\rho-\sigma}),
    \end{align}
    where $d$ is the dimension of the Hilbert space and $\eta(x):= -x\cdot \log(x)$.
\end{lemma}

We give the proof's construction in \Cref{sscn:gse_construction}, its completeness analysis in \Cref{sscn:gse_YES}, soundness analysis in \Cref{sscn:gse_NO}, and approximation ratio analysis in \Cref{sscn:approxgse}. 

\subsection{Construction}\label{sscn:gse_construction}


As in the proof of \Cref{thm:GSCON}, let $(V,g,g')$ be an instance of QMSA of size $N$, with $V$ taking in an $n$-qubit proof in register $B$ and with ancilla register $C$, and let $\LL=\abs{V}$.
Without loss of generality, the completeness and soundness parameters for $V$ are $1-\varepsilon$ and $\varepsilon$ for $\varepsilon \in \Theta(1/2^{N})$~\cite{gharibianHardnessApproximationQuantum2012}. We map $(V,g,g')$ to an instance of GSE as follows. 

We begin by introducing four registers $E,E',F$, and $F'$ of size $n$ each, set initially to the all zeroes state. These registers will be used for two operations: First, we copy the proof in $B$ to registers $E$ and $E'$. Second, we create Bell pairs on registers $F$ and $F'$, depending on whether the proof qubits in $B$ are $0$ or $1$. For these two operations, define $3$-qubit gates $S_1$ and $S_2$ with action
\begin{align}
    &S_1(\ket{000}) = \ket{000} \quad \text{ and }\quad S_1(\ket{100}) = \ket{111},\\
    &S_2(\ket{000}) = \ket{000} \quad \text{ and }\quad S_2(\ket{100}) = \frac{1}{\sqrt{2}}\ket{100} + \frac{1}{\sqrt{2}}\ket{111}. 
\end{align}
Note that both gates are controlled on the first qubit. These gates are 3-local, but can be simulated with a constant number of 2-local gates. So, define new circuit $W = W_1 \cdots W_{2(\LL+2n)}$ as follows. First, add $n$ $S_1$ gates, where each gate is controlled by one of the proof qubits $1\le i \le n$, and acts on registers $E_{i}$ and $E'_{i}$, i.e. on $B_iE_iE'_i$. Similarly, we create Bell pairs controlled by the proof register. For this we apply $n$ $S_2$ gates, again controlled by the proof qubits, acting on $F_i$ and $F'_i$, i.e. on $B_iF_iF'_i$. $W$ next runs $V$ on registers $B$ and $C$. Lastly, append to $W$ $2n+\LL$ identity gates. In total, our new circuit has $\LL':= 4n+2\LL$ many gates. 

The input for our GSE instance is now: $H:=\mu(\hin+\hprop+\hstab)+\hout$ (for $\mu$ to be set later) is the Kitaev Hamiltonian of $W$ (\Cref{l:kitaev}) acting on registers $B$ (proof), $C$ (ancilla), $D$ (clock).  Define $\Tilde{n}$ to be the number of qubits the new Hamiltonian acts on. We set $\eta_1 := \alpha$, $\eta_2 := \beta/2\Tilde{n}^3$, $\eta_3 := g+1$, $\eta_4 := \frac{1}{4}(g'-3)$. The partition $(A,A')$ contains registers $E$ and $F$ on one side and the remaining registers on the other.

\subsection{YES case}\label{sscn:gse_YES}
We have to show that there exists a state $\ket\psi$ with energy $\le \alpha$ against $H$ and with entanglement entropy $\le g+1$ across cut $(A,A')$. In the YES case for QMSA, we know that circuit $V$ accepts some proof $x^* \in \{0,1\}^n$ of Hamming weight at most $g$. Since $V$ accepts a monotone set by definition, without loss of generality, $x^*$ has weight exactly $g$. Since the circuit accepts the proof, by \Cref{l:kitaev} the history state
\begin{align}
    \ket\psi := \frac{1}{\sqrt{\LL'+1}}\sum_{t=0}^{\LL'} W_t \cdots W_1 \ket{x^*}_B\ket{0\cdots 0}_C\ket{t}_D\ket{0\cdots 0}_{EE'FF'}=:\frac{1}{\sqrt{\LL'+1}}\sum_{t=0}^{\LL'} \ket{\psi_t}_{BCD}\ket{\phi_t}_{EE'FF'}
\end{align}
has energy $\bra\psi H \ket\psi\le \alpha$.

We next show $\ket\psi$ has low entanglement entropy by looking at its Schmidt decomposition across $(A,A')$. Let $R_{x^*} = \{i \in \{1,\ldots,n\}: x^*_i = 1\}$ and $R_{x^*}' = \{x \in \{0,1\}^n: \{i\in [n]: x_i = 1\} \subseteq R_{x^*}\}$. After timestep $t \ge 2n$, 
\begin{align}
    W_t \cdots W_1 \ket{x^*}_B\ket{0\cdots 0}_C\ket{t}_D\ket{0\cdots 0}_{EE'FF'} = \ket{\psi_t}_{BCD}\ket{x^*}_E\ket{x^*}_{E'}\ket\phi_{FF'},
\end{align}
where 
\begin{align}
    \ket\phi = \sum\limits_{x \in R_{x^*}'}\frac{1}{\sqrt{2^g}} \ket{x}_F\ket{x}_{F'}.
\end{align}
Since we were able to write down $\ket\phi$ in its Schmidt decomposition, we are able to calculate its entanglement directly as
\begin{align}
    \s(\trace_{F'}(\ketbra{\phi}{\phi})) = -\sum\limits_{x \in R_{x^*}'}\frac{1}{2^g} \log\left(\frac{1}{2^g}\right) = -2^g \cdot \frac{1}{2^g} \log\left(\frac{1}{2^g}\right) = g.
\end{align}
This gives the entanglement entropy for $\ket\phi_{FF'}$, but we want
the entropy for $\ket\psi$. For this, since the $S_1$ gates act
simultaneously on registers $E_i$ and $E'_i$, we have

\begingroup
\allowdisplaybreaks[4]
\begin{align}
    \ket\psi
    &= \frac{1}{\sqrt{\LL'+1}}\sum_{t=0}^{\LL'}
       W_t \cdots W_1 \ket{x^*}_B\ket{0\cdots 0}_C\ket{t}_D
       \ket{0\cdots 0}_{EE'FF'} \\
    &= \frac{1}{\sqrt{\LL'+1}}\sum_{t=0}^{\LL'}
       \ket{\psi_t}\ket{w_t}_{E}\ket{w_t}_{E'}
       \sum_{x \in \{0,1\}^n} a_{t,x}\ket{x}_{F}\ket{x}_{F'}
       \notag\\[-2pt]
    &\qquad
       \text{(for some $w_t\in \set{0,1}^n$ and amplitudes $a_{t,x}$)} \\
    &= \frac{1}{\sqrt{\LL'+1}}
       \sum_{t=0}^{n} \ket{\psi_t}\ket{w_t}_{E}\ket{w_t}_{E'}
       \sum_{x \in \{0,1\}^n} a_{t,x}\ket{x}_{F}\ket{x}_{F'}
       \notag\\
    &\quad + \frac{1}{\sqrt{\LL'+1}}
       \sum_{t=n+1}^{\LL'} \ket{\psi_t}\ket{x^*}_{E}\ket{x^*}_{E'}
       \sum_{x \in \{0,1\}^n} a_{t,x}\ket{x}_{F}\ket{x}_{F'} \\
    &= \frac{1}{\sqrt{\LL'+1}}
       \sum_{t=0}^{n} a_{t,0^n}\ket{\psi_t}\ket{w_t}_{E}\ket{w_t}_{E'}
       \ket{0^n}_{F}\ket{0^n}_{F'}
       \notag\\
    &\quad + \frac{1}{\sqrt{\LL'+1}}
       \sum_{t=n+1}^{\LL'} \ket{\psi_t}\ket{x^*}_E\ket{x^*}_{E'}
       \sum_{x \in R_{x^*}'} a_{t,x}\ket{x}_{F}\ket{x}_{F'}.
\end{align}
\endgroup

To calculate the entanglement entropy via the Schmidt decomposition, we look at the weights on $\ket{y}_E\ket{y}_{E'}\ket{x}_F\ket{x}_{F'}$. 
Specifically, recall the $(A,A')$ cut partitions registers $EF$ versus all other registers. So we consider the Schmidt decomposition with basis elements of form $\ket{y}_E\ket{x}_{F}$ for $x\in R'_{x^*}$ (all other $x$ correspond to terms of amplitude $0$, and thus can be ignored). 

Now, for any  $w_t \ne x^*$, we have only non-zero weight on $\ket{w_t}_E\ket{w_t}_{E'}\ket{0^n}_F\ket{0^n}_{F'}$. Since $x^*$ has Hamming weight $g$, there are at most $g$ such states. The only other states with non-zero weight have the form $\ket{x^*}_E\ket{x^*}_{E'}\ket{x}_F\ket{x}_{F'}$, for $x \in R_{x^*}'$. Since $|R_{x^*}'| \le 2^g$, there are only $2^g$ such states. In total, we have only $2^g + g\le 2^{g+1}$ non-zero amplitudes, so the entropy is at most $g+1$ across the $(A,A')$ cut.

\subsection{NO case}\label{sscn:gse_NO}

We have to show that every state with energy $\le \beta$ has entanglement entropy $\ge \frac{1}{4}(g'-3)$. Consider first the special case of any history state
\begin{align}
    \ket\psi := \sum_{t=0}^{\LL'} \frac{1}{\sqrt{\LL'+1}} W_t \cdots W_1 \ket{\gamma}_B\ket{0\cdots 0}_C\ket{t}_D\ket{0\cdots 0}_{EE'FF'},
\end{align}
where $\ket\gamma$ is a superposition of standard basis states with Hamming weight at least $g'$. Let $\gamma = \sum_{p \in \{0,1\}^n} b_p \ket{p}$ for amplitudes $b_p$. We can rewrite
\begin{align}
    \ket\psi &= 
    \sum_{t=0}^{\LL'} \frac{1}{\sqrt{\LL'+1}} \sum_{p\in \{0,1\}^n} {b_p}\cdot W_t \cdots W_1 \ket{p}_B\ket{0\cdots 0}_C\ket{t}_D\ket{0\cdots 0}_{EE'FF'} \\
    &= \frac{1}{\sqrt{\LL'+1}} \sum_{p\in \{0,1\}^n}  {b_p} \left(\sum_{x,\; p' \in \{0,1\}^n} \sum_{t=0}^{\LL'}  a_{t,p,p',x} \ket{\psi_{t,p}'} \ket{p'}_E\ket{p'}_{E'}\ket{x}_F\ket{x}_{F'}\right). 
\end{align}
    The $a_{t,p,p',x}$ can be interpreted as follows: Having a history state with $p$ in its proof register, at time $t$ we have weight $|a_{t,p,p',x}|^2$ on $\ket{p'}_E\ket{p'}_{E'}\ket{x}_F\ket{x}_{F'}$. 
    For fixed $p$ of Hamming weight $\HW(p)$, after a certain timestep we have an equal superposition over the $x \in R_p'$ on register $F$, which is a superposition over $2^{HW(p)}$ many elements. This certainly holds for the last half of the timesteps, since the last half of the gates in $W$ are identity gates. But at that point, all the weight on registers $EE'$ is on $\ket{p}_E\ket{p}_{E'}$. Therefore we have $\sum _{t=0}^{\LL'}  |a_{t,p,p,x}|^2 \ge \frac{\LL'+1}{2} \frac{1}{2^{\HW(p)}}$. Let $j\in\set{1,\ldots, n}$ be the largest index such that $p_j = 1$. Define 
    \begin{align}
            M_p := \{x \in R_p' : x_j = 1\}.
    \end{align} 
    For $t \ge 2n$, we have an equal superposition over all the $2^{\HW(p)}$ elements in $R'_p$ on registers $F$ and $F'$. The elements in $M_p$ are exactly those where $|a_{t,p,p,x}|$ is always 0 or $1/2^{\HW(p)}$, depending on $t$. Therefore, for $x \in M_p$, we have 
    \begin{align}\label{GSEinequal}
        \frac{\LL'+1 }{2^{\HW(p)+1}} \le \sum _{t=0}^{\LL'}  |a_{t,p,p,x}|^2 \le (\LL'+1) \frac{1}{2^{\HW(p)}}.
    \end{align}
Let $f(x):= -x\log(x)$. To lower bound the entanglement entropy of $\ket\psi$ across the $(A,A')$ cut, consider an $(A,A')$ Schmidt decomposition with Schmidt basis $\ket{p'}_E\ket{x}_F$ on $A$. 
Then, for any $p$ and $x$, the corresponding Schmidt coefficient is at least
\begin{align}\label{eqn:schmidt}
    \sqrt{\frac{\abs{b_p}^2\sum_{t}\abs{a_{t,p,p,x}}^2}{\LL'+1}},
\end{align}
where we have used the fact that the $\ket{\psi'_{t,p'}}$ are orthogonal, since (1) history states based on distinct proof strings $p$ are orthogonal, and (2) the clock states on register $D$ are orthogonal. We thus have
\begin{align}
    \s(\trace_{A'}(\ketbra{\psi}{\psi})) &\ge \sum_{p} \sum_{x \in \{0,1\}^n} f\left(\frac{\sum_{t}\abs{b_p}^2\abs{a_{t,p,p,x}}^2}{\LL'+1}\right) \\
    &\ge  \sum_{p} \sum_{x \in \set{0,1}^n} \abs{b_p}^2 \cdot f\left(\frac{\sum_{t}\abs{a_{t,p,p,x}}^2}{\LL'+1}\right) \\
    &\ge  \sum_{p} \sum_{x \in M_p}  \abs{b_p}^2 \cdot f\left(\frac{\sum_{t}\abs{a_{t,p,p,x}}^2}{\LL'+1}\right) \\
    &\ge  \sum_{p} \sum_{x \in M_p}  \abs{b_p}^2 \cdot f\left(\frac{1}{2^{\HW(p)+1}}\right) \\
    &=  \sum_p \abs{b_p}^2 \sum_{x \in M_p} \frac{1}{2^{\HW(p)+1}}\cdot (\HW(p)+1) \\ 
    &=  \sum_p \abs{b_p}^2 \cdot \frac{1}{4}\left(\HW(p)+1\right) \\
    &\ge  \sum_p \abs{b_p}^2 \cdot \frac{1}{4}\left(g'+1\right)  \\
    &\ge \frac{1}{4} (g'+1),
\end{align}
where the second inequality follows by concavity of entropy, the third because $f(x)\geq 0$ on $[0,1]$, the fourth by \Cref{eqn:schmidt} and since $f(x)$ is concave for small $x$, the fifth since $\abs{M_p}=2^{\HW(p)-1}$, the sixth since $\HW(p)\geq g'$ by assumption, and the last since the $\abs{b_p}^2$ form a distribution.

The analysis above holds for any history state. We now extend the analysis to general low-energy states by combining the Extended Projection Lemma~\Cref{l:kkr} with the Fannes inequality \Cref{eqn:Fannes}.
Specifically, suppose $\bra{\psi}\mu(\hin+\hprop+\hstab)+\hout\ket{\psi}\leq \eta_2 = \beta/2\Tilde{n}^3$.
We again use the Extended Projection Lemma (\Cref{l:kkr}) with $H_1=\mu(\hin+\hprop+\hstab)$ and $H_2=\hout$, where $\spa{S}$ denotes the null space of $\hin+\hprop+\hout$. Combining this with \Cref{l:GKgap}, we conclude that there exists a history state $\ket{\psihist'}$ such that 
\begin{align}
    \abs{\braket{\psi}{\psihist'}}^2\geq 1-\frac{2\pi^4}{\mu 2\Tilde{n}^3}.
\end{align}
A history state with weight $a$ on low Hamming weight proofs will have energy $\ge a \cdot \beta$. We have $\eta_2 = 1/2\Tilde{n}^3 \cdot \beta$. Therefore, by choosing $\mu$ large enough polynomial, a low energy state must have overlap $\ge 1-1/\Tilde{n}^3$ with a history state $\ket{\psihist}$ with only high Hamming weight proofs in its proof register.
So by \Cref{eqn:enorm} we have
\begin{align}\label{eqn:trnorm2}
    \trnorm{\ketbra{\psi}{\psi}-\ketbra{\psihist}{\psihist}}\leq \frac{2}{\sqrt{\Tilde{n}^3}}.
\end{align}
By the Fannes inequality~\Cref{eqn:Fannes}, we conclude
\begin{align}
    \abs{S(\ketbra{\psi}{\psi})-S(\ketbra{\psihist}{\psihist})}\leq \frac{2\Tilde{n}}{\sqrt{\Tilde{n}^3}} - \frac{2}{\sqrt{\Tilde{n}^3}} \cdot \log\left({\frac{2}{\sqrt{\Tilde{n}^3}}}\right)\le 1.
\end{align}
This allows us to lower bound the entanglement of any low energy state as
$$S(\ketbra{\psi}{\psi}) \ge S(\ketbra{\psihist}{\psihist}) - 1 \ge \frac{1}{4}(g'+1)-1 = \frac{1}{4}(g'-3).$$
 (Note $N\in\poly(g')$ for $N$ the input size of the QMSA instance.)



\subsection{Approximation ratio}\label{sscn:approxgse}

Fix any $\epsilon>0$. Recall that $N$ is the encoding size of the QMSA instance with ratio $g'/g\in \Theta(N^{1-\epsilon})$, and let $N'$ denote the encoding size of our constructed GSE instance. We wish to show that $\eta_4/\eta_3\in\Theta(N'^{1-\epsilon})$. For this, note first that 
\begin{align}
    \frac{\eta_4}{\eta_3}=\frac{\frac{1}{4}(g'-3)}{g+1}\in\Theta\left(N^{1-\epsilon}\right).
\end{align}
Thus, the claim will follow if $N'\in\Theta(N)$. The latter holds because our construction maps the QMSA circuit $V$ to a new circuit of size $W$, to which it applies the Kitaev circuit-to-Hamiltonian construction. Recalling that  $\abs{W}=2T+4n$ for $T,n\in O(N)$, we conclude $N'\in\Theta(N)$.

\section{NP-hardness of approximation}

In this section, we derive a hardness of approximation result for the classical analogue of GSCON, Boolean Reconfiguration. Again, we do this by reduction from a Hamming weight problem that is known to be hard to approximate:

\begin{definition}[Monotone Minimum Satisfying Assignment (MMSA)~\cite{umansHardnessApproximatingSpl1999}]\label{def:MMSA}
Given a monotone classical circuit $C$ and two thresholds $g,g'$, decide whether there exists an input of Hamming weight $\le g$ that is accepted, or whether all accepted inputs have Hamming weight $\ge g'$.
\end{definition}

\begin{theorem}[\cite{umansHardnessApproximatingSpl1999}]\label{def:MMSA_is_NPhard}
    MMSA is NP hard for $\frac{g'}{g} \in \Theta(N^{1/5-\epsilon})$.
\end{theorem}
\noindent Combining this with the improved dispersers constructed in \cite{a.ta-shmaLosslessCondensersUnbalanced2007} immediately yields~\cite{umansHardnessApproximatingSpl1999} the stronger result:
\begin{theorem}
    MMSA is NP hard for $\frac{g'}{g} \in \Theta(N^{1-\epsilon})$.
\end{theorem}

\begin{definition}[Boolean Reconfiguration (BR)]\label{def:BR}
    Given a $k$-SAT formula $\phi$ with $n$ variables, two solutions $s,t \in \{0,1\}^n$ and thresholds $h$ and $h'$, does there exist a path of satisfying assignments $s=s_1,\ldots,s_h=t$, such that each $s_i$ can be obtained from $s_{i-1}$ by flipping a single bit, or does every such path have length $\ge h'$.
\end{definition}

\thmBool*

\begin{proof}
    Reduction from MMSA: We are given a circuit $C= C_1,\ldots,C_m$ acting on variables $x_1,\ldots,x_n$ and thresholds $g$ and $g'$. We define $n+m$ variables $y_1, \ldots, y_{n+m}$, where the first $n$ variables correspond to the bits $C$ acts on, and the last $m$ variables correspond to the $m$ gates. We use the gates of $C$ to define constraints as follows:
    \begin{enumerate}\label{constraints_classical}
        \item For every NOT gate with input $y_i$ and output $y_j$ we put the constraint
        \begin{equation}\label{not_gates_formula} (y_i \vee y_j) \wedge (\neg y_i \vee \neg y_j). \end{equation}
        \item For an AND gate with inputs $y_i, y_j$ and output $y_k$ we define 
        \begin{equation} (\neg y_k \vee y_i) \wedge (\neg y_k \vee y_j) \wedge (\neg y_i \vee \neg y_j \vee y_k).\end{equation}
        \item For an OR gate with inputs $y_i, y_j$ and output $y_k$ we define
        \begin{equation}\label{OR_gates_formula} (y_k \vee \neg y_i) \wedge (y_k \vee \neg y_j) \wedge (y_i \vee y_j \vee \neg y_k).\end{equation}
        \item For the output gate we put a single constraint \begin{equation}\label{accept_gate_formula}(y_{n+m}). \end{equation} 
    \end{enumerate}
    The constructed formula is satisfiable iff there exists an input that $C$ accepts. A satisfying assignment would be setting the variables $y_1,\ldots,y_n$ according to an input $C$ accepts, and setting $y_{n+1},\ldots,y_{n+m}$ according to the outputs of the gates in $C$. For the amplification, we add variables and constraints analogously to the GSCON setting: We add three GO variables $e_1,e_2,e_3$, variables $f_1,\ldots,f_n$ and $g_1,\ldots,g_n$ to copy the input bits, clock variables $k_1,\ldots,k_{4n}$ and $l_1,l_2$ and amplification variables $a_1,\ldots, a_{m+n}$. The starting and target states are on the one hand the all 0's string, and on the other the all 0's string except bits $e_2$ and $e_3$ which are 1. To get from $(e_2,e_3) = (0,0)$ to $(1,1)$, we need to go to either (1,0) or (0,1). We define constraints to make this only possible when $e_1$ is flipped to 1:
    \begin{equation}(e_1 \vee e_2 \vee \neg e_3) \wedge (e_1 \vee \neg e_2 \vee e_3).\end{equation}
    Now we modify the constraints in \ref{not_gates_formula}-\ref{accept_gate_formula} to only be active when $e_1 = 1$. We do this by adding a $\neg e_1$ to all constraints in \ref{not_gates_formula}-\ref{accept_gate_formula}. For example, for an OR gate this would mean
    \begin{align} &(y_k \vee \neg y_i) \wedge (y_k \vee \neg y_j) \wedge (y_i \vee y_j \vee \neg y_k) \\
    \mapsto  &(\neg e_1 \vee y_k \vee \neg y_i) \wedge (\neg e_1 \vee y_k \vee \neg y_j) \wedge (\neg e_1 \vee y_i \vee y_j \vee \neg y_k). \end{align}
    These constraints are the classical analogs to the constraints in $H_{Kit} \otimes P$ in the GSCON setting. We further substitute $P \otimes \ketbra{0}{0}_{K_{4n}}$ with $(\neg e_1 \vee k_{4n})$. Lastly, we substitute the constraints from a) to f). They all have the same form: Either $\ketbra{ab}{ab}$ or $I-\ketbra{ab}{ab}$, sometimes connected via an $\otimes$. We transform these into boolean constraints by mapping
    \begin{equation}
        \ketbra{00}{00} \mapsto (z_i \vee z_j)
    \end{equation}
    or
    \begin{equation}
        I-\ketbra{00}{00} \mapsto (\neg z_i \wedge \neg z_j),
    \end{equation}
    where $z_i$ and $z_j$ are the variables corresponding to the registers of $\ket{00}$. In case $a$ or $b$ is 1 rather than 0, change the constraint by negating the corresponding variables. If two hermitian constraints are connected by a $\otimes$, put a $\vee$ between the two boolean constraints. We do this transformation for all constraints a) to f) except b), as we have to account for the fact that we only can change one bit of the clock register at a time in the classical setting. In the GSCON setting, to go from a timestep $t_i$ to $t_{i+1}$, one needs to apply X gates one qubit in K and one in L simultaneously. In the classical setting, we will allow the verifier to first change the bit in register K, and then the one in register L. The new clock advances as follows:
    \begin{align}
&\{0,0,0,0,...0\}_K\{0,0\}_L\\
&\{1,0,0,0,...0\}_K \{0,0\}_L\\
&\{1,0,0,0,...0\}_K\{1,0\}_L\\
&\{1,1,0,0,...0\}_K\{1,0\}_L\\
&\{1,1,0,0,...0\}_K\{1,1\}_L\text{, etc}\ldots
\end{align}
This means if $\{\widetilde{i}\}_K\{0,0\}_L$ is a valid timestep, so is $\{\widetilde{i}\}_K\{0,1\}_L$, same for timesteps $\{\widetilde{j}\}_K\{1,0\}_L$ and $\{\widetilde{j}\}_K\{0,0\}_L$, and so on.  So rather than transforming
    \begin{equation}
        \ketbra{10}{10}_{K_{4i-1},K_{4i}} \otimes (I-\ketbra{00}{00})_L,
    \end{equation}\\
    into 
    \begin{equation}
        (\neg k_{4i-1} \vee k_{i}) \vee (\neg l_i \wedge \neg l_j),
    \end{equation}
    we transform it into
    \begin{equation}
        (\neg k_{4i-1} \vee k_{i}) \vee (\neg l_i \wedge \neg l_j) \vee (\neg l_{i} \wedge l_{j}),
    \end{equation}\\
    and do the analogous transformation with the constraints where $\ket{ab}_L \ne \ket{00}$. \\
    We set the thresholds 
    \begin{align}
    h := 2\cdot(3g + 2g(m+n) + 4n + 4n) + 4, 
    \qquad\text{and}\qquad
    h' := g' \cdot (m+n).
\end{align}

    Completeness and soundness for the construction follow completely analogously from the completeness of the GSCON construction and Lemma \ref{l:SoundnessGSCONXGates}: To go from $(e_2,e_3) = (0,0)$ to $(1,1)$, at some point $(e_2,e_3) = (1,0)$ or $(e_2,e_3) = (0,1)$. For this to happen, we need to flip $e_1$ to 1, which activates the constraints that check if the assignment $y_1,\ldots,y_{n+m}$ satisfies the constraints obtained from $C$. In the completeness case, preparing an accepted assignment on $y_1,\ldots,y_{n+m}$ takes $3g + 2g(m+n) + 4n + 4n$ many bit flips, as we need to prepare the three proof copies $g$ times, flip the amplification bits to 1 and back, and go through $4n$ timesteps on registers $K$ and $L$. To get from $(e_1,e_2,e_3) + (0,0,0)$ to $(0,1,1)$, we need to flip 4 bits, and then again $3g + 2g(m+n) + 4n + 4n$ many flips to uncompute everything. For soundness, in order to prepare an accepting proof, we need to flip all bits in the amplification register at least g' times, so preparing the target state is impossible in $h'$ many steps. 
\end{proof}

\bibliographystyle{alpha}
\bibliography{literature.bib,Sev.bib}

@article{a.ta-shmaLosslessCondensersUnbalanced2007,
  title = {Lossless {{Condensers}}, {{Unbalanced Expanders}}, and {{Extractors}}},
  author = {{A. Ta-Shma} and {C. Umans} and {D. Zuckerman}},
  year = 2007,
  journal = {Combinatorica},
  volume = {27},
  number = {2},
  pages = {213--240}
}

@article{agarwalOracleSeparationsQuantumClassical2024,
  title = {Oracle {{Separations}} for the {{Quantum-Classical Polynomial Hierarchy}}},
  author = {Agarwal, Avantika and {Ben-David}, Shalev},
  year = 2024,
  number = {arXiv:2410.19062},
  eprint = {2410.19062},
  publisher = {arXiv},
  doi = {10.48550/arXiv.2410.19062},
  url = {http://arxiv.org/abs/2410.19062},
  urldate = {2024-11-07},
  archiveprefix = {arXiv},
  journal = {10.48550/arXiv.2410.19062}
}

@inproceedings{agarwalQuantumPolynomialHierarchies2024,
  title = {Quantum {{Polynomial Hierarchies}}: {{Karp-Lipton}}, {{Error Reduction}}, and {{Lower Bounds}}},
  shorttitle = {Quantum {{Polynomial Hierarchies}}},
  booktitle = {49th {{International Symposium}} on {{Mathematical Foundations}} of {{Computer Science}} ({{MFCS}} 2024)},
  author = {Agarwal, Avantika and Gharibian, Sevag and Koppula, Venkata and Rudolph, Dorian},
  year = 2024,
  publisher = {Schloss Dagstuhl -- Leibniz-Zentrum f\"ur Informatik},
  doi = {10.4230/LIPIcs.MFCS.2024.7},
  url = {https://drops.dagstuhl.de/entities/document/10.4230/LIPIcs.MFCS.2024.7},
  urldate = {2024-08-30},
  copyright = {https://creativecommons.org/licenses/by/4.0/legalcode},
  langid = {english}
}

@inproceedings{aharonovComplexityCommutingLocal2011,
  title = {On the {{Complexity}} of {{Commuting Local Hamiltonians}}, and {{Tight Conditions}} for {{Topological Order}} in {{Such Systems}}},
  booktitle = {Proceedings of the 2011 {{IEEE}} 52nd {{Annual Symposium}} on {{Foundations}} of {{Computer Science}}},
  author = {Aharonov, Dorit and Eldar, Lior},
  year = 2011,
  series = {{{FOCS}} '11},
  pages = {334--343},
  publisher = {IEEE Computer Society},
  address = {USA},
  doi = {10.1109/FOCS.2011.58},
  url = {https://doi.org/10.1109/FOCS.2011.58},
  urldate = {2024-08-30},
  isbn = {978-0-7695-4571-4}
}

@inproceedings{aharonovComplexityTwoDimensional2018,
  title = {On the {{Complexity}} of {{Two Dimensional Commuting Local Hamiltonians}}},
  booktitle = {13th {{Conference}} on the {{Theory}} of {{Quantum Computation}}, {{Communication}} and {{Cryptography}} ({{TQC}} 2018)},
  author = {Aharonov, Dorit and Kenneth, Oded and Vigdorovich, Itamar},
  editor = {Jeffery, Stacey},
  year = 2018,
  series = {Leibniz {{International Proceedings}} in {{Informatics}} ({{LIPIcs}})},
  volume = {111},
  pages = {2:1--2:21},
  publisher = {Schloss Dagstuhl--Leibniz-Zentrum fuer Informatik},
  address = {Dagstuhl, Germany},
  issn = {1868-8969},
  doi = {10.4230/LIPIcs.TQC.2018.2},
  url = {http://drops.dagstuhl.de/opus/volltexte/2018/9249},
  urldate = {2023-10-12},
  isbn = {978-3-95977-080-4}
}

@article{aharonovGuestColumnQuantum2013,
  title = {Guest {{Column}}: {{The Quantum PCP Conjecture}}},
  author = {Aharonov, Dorit and Arad, Itai and Vidick, Thomas},
  year = 2013,
  journal = {SIGACT News},
  volume = {44},
  number = {2},
  pages = {47--79},
  publisher = {ACM},
  address = {New York, NY, USA},
  issn = {0163-5700},
  doi = {10.1145/2491533.2491549},
  url = {http://doi.acm.org/10.1145/2491533.2491549}
}

@article{aharonovPowerQuantumSystems2009,
  title = {The {{Power}} of {{Quantum Systems}} on a {{Line}}},
  author = {Aharonov, Dorit and Gottesman, Daniel and Irani, Sandy and Kempe, Julia},
  year = 2009,
  journal = {Communications in Mathematical Physics},
  volume = {287},
  number = {1},
  pages = {41--65},
  issn = {1432-0916},
  doi = {10.1007/s00220-008-0710-3},
  url = {https://doi.org/10.1007/s00220-008-0710-3},
  urldate = {2023-01-31},
  langid = {english}
}

@inproceedings{aharonovQuantumNPSurvey2002,
  title = {Quantum {{NP}} - {{A Survey}}},
  author = {Aharonov, Dorit and Naveh, Tomer},
  year = 2002,
  eprint = {quant-ph/0210077},
  publisher = {arXiv:quant-ph/0210077},
  doi = {10.48550/arXiv.quant-ph/0210077},
  url = {http://arxiv.org/abs/quant-ph/0210077},
  urldate = {2023-10-11},
  archiveprefix = {arXiv}
}

@inproceedings{anshuNLTSHamiltoniansGood2023,
  title = {{{NLTS Hamiltonians}} from {{Good Quantum Codes}}},
  booktitle = {Proceedings of the 55th {{Annual ACM Symposium}} on {{Theory}} of {{Computing}}},
  author = {Anshu, Anurag and Breuckmann, Nikolas P. and Nirkhe, Chinmay},
  year = 2023,
  series = {{{STOC}} 2023},
  pages = {1090--1096},
  publisher = {Association for Computing Machinery},
  address = {New York, NY, USA},
  doi = {10.1145/3564246.3585114},
  url = {https://doi.org/10.1145/3564246.3585114},
  urldate = {2023-09-21},
  isbn = {978-1-4503-9913-5}
}

@inproceedings{aradLinearTimeAlgorithm2016,
  title = {Linear {{Time Algorithm}} for {{Quantum 2SAT}}},
  booktitle = {43rd {{International Colloquium}} on {{Automata}}, {{Languages}}, and {{Programming}} ({{ICALP}} 2016)},
  author = {Arad, Itai and Santha, Miklos and Sundaram, Aarthi and Zhang, Shengyu},
  editor = {Chatzigiannakis, Ioannis and Mitzenmacher, Michael and Rabani, Yuval and Sangiorgi, Davide},
  year = 2016,
  series = {Leibniz {{International Proceedings}} in {{Informatics}} ({{LIPIcs}})},
  volume = {55},
  pages = {15:1--15:14},
  publisher = {Schloss Dagstuhl--Leibniz-Zentrum fuer Informatik},
  address = {Dagstuhl, Germany},
  issn = {1868-8969},
  doi = {10.4230/LIPIcs.ICALP.2016.15},
  url = {http://drops.dagstuhl.de/opus/volltexte/2016/6279},
  urldate = {2023-10-26},
  isbn = {978-3-95977-013-2}
}

@article{bauschComplexityTranslationallyInvariant2017,
  title = {The {{Complexity}} of {{Translationally Invariant Spin Chains}} with {{Low Local Dimension}}},
  author = {Bausch, Johannes and Cubitt, Toby and Ozols, Maris},
  year = 2017,
  journal = {Annales Henri Poincar\'e},
  volume = {18},
  number = {11},
  pages = {3449--3513},
  issn = {1424-0661},
  doi = {10.1007/s00023-017-0609-7},
  url = {https://doi.org/10.1007/s00023-017-0609-7},
  urldate = {2023-10-11},
  langid = {english}
}

@inproceedings{beaudrapLinearTimeAlgorithm2016,
  title = {A {{Linear Time Algorithm}} for {{Quantum}} 2-{{SAT}}},
  booktitle = {31st {{Conference}} on {{Computational Complexity}} ({{CCC}} 2016)},
  author = {{\noopsort{beaudrap}}de Beaudrap, Niel and Gharibian, Sevag},
  editor = {Raz, Ran},
  year = 2016,
  series = {Leibniz {{International Proceedings}} in {{Informatics}} ({{LIPIcs}})},
  volume = {50},
  pages = {27:1--27:21},
  publisher = {Schloss Dagstuhl--Leibniz-Zentrum fuer Informatik},
  address = {Dagstuhl, Germany},
  issn = {1868-8969},
  doi = {10.4230/LIPIcs.CCC.2016.27},
  url = {http://drops.dagstuhl.de/opus/volltexte/2016/5836},
  urldate = {2023-10-27},
  isbn = {978-3-95977-008-8}
}

@inproceedings{ben-aroyaQuantumExpandersMotivation2008,
  title = {Quantum {{Expanders}}: {{Motivation}} and {{Constructions}}},
  shorttitle = {Quantum {{Expanders}}},
  booktitle = {2008 23rd {{Annual IEEE Conference}} on {{Computational Complexity}}},
  author = {{Ben-Aroya}, Avraham and Schwartz, Oded and {Ta-Shma}, Amnon},
  year = 2008,
  pages = {292--303},
  issn = {1093-0159},
  doi = {10.1109/CCC.2008.23},
  url = {https://ieeexplore.ieee.org/document/4558831},
  urldate = {2024-08-29}
}

@inproceedings{bittelOptimalDepthVariational2023,
  title = {The {{Optimal Depth}} of {{Variational Quantum Algorithms Is QCMA-Hard}} to {{Approximate}}},
  booktitle = {38th {{Computational Complexity Conference}} ({{CCC}} 2023)},
  author = {Bittel, Lennart and Gharibian, Sevag and Kliesch, Martin},
  editor = {{Ta-Shma}, Amnon},
  year = 2023,
  series = {Leibniz {{International Proceedings}} in {{Informatics}} ({{LIPIcs}})},
  volume = {264},
  pages = {34:1--34:24},
  publisher = {Schloss Dagstuhl -- Leibniz-Zentrum f\"ur Informatik},
  address = {Dagstuhl, Germany},
  issn = {1868-8969},
  doi = {10.4230/LIPIcs.CCC.2023.34},
  url = {https://drops.dagstuhl.de/opus/volltexte/2023/18304},
  urldate = {2023-09-21},
  isbn = {978-3-95977-282-2}
}

@inproceedings{boulandPublicKeyPseudoentanglementHardness2024,
  title = {Public-{{Key Pseudoentanglement}} and the {{Hardness}} of {{Learning Ground State Entanglement Structure}}},
  booktitle = {39th {{Computational Complexity Conference}} ({{CCC}} 2024)},
  author = {Bouland, Adam and Fefferman, Bill and Ghosh, Soumik and Metger, Tony and Vazirani, Umesh and Zhang, Chenyi and Zhou, Zixin},
  year = 2024,
  publisher = {Schloss Dagstuhl -- Leibniz-Zentrum f\"ur Informatik},
  doi = {10.4230/LIPIcs.CCC.2024.21},
  url = {https://drops.dagstuhl.de/entities/document/10.4230/LIPIcs.CCC.2024.21},
  urldate = {2024-08-29},
  copyright = {https://creativecommons.org/licenses/by/4.0/legalcode},
  langid = {english}
}

@article{bravyiCommutativeVersionLocal2005,
  title = {Commutative Version of the Local {{Hamiltonian}} Problem and Common Eigenspace Problem},
  author = {Bravyi, Sergey and Vyalyi, Mikhail},
  year = 2005,
  journal = {Quantum Information \& Computation},
  volume = {5},
  number = {3},
  pages = {187--215},
  issn = {1533-7146}
}

@article{bravyiEfficientAlgorithmQuantum2006,
  title = {Efficient Algorithm for a Quantum Analogue of 2-{{SAT}}},
  author = {Bravyi, Sergey},
  year = 2006,
  eprint = {quant-ph/0602108},
  publisher = {arXiv:quant-ph/0602108},
  doi = {10.48550/arXiv.quant-ph/0602108},
  url = {http://arxiv.org/abs/quant-ph/0602108},
  urldate = {2023-09-06},
  archiveprefix = {arXiv},
  journal = {10.48550/arXiv.quant-ph/0602108}
}

@inproceedings{cookComplexityTheoremprovingProcedures1971,
  title = {The Complexity of Theorem-Proving Procedures},
  booktitle = {Proceedings of the Third Annual {{ACM}} Symposium on {{Theory}} of Computing},
  author = {Cook, Stephen A.},
  year = 1971,
  series = {{{STOC}} '71},
  pages = {151--158},
  publisher = {Association for Computing Machinery},
  address = {New York, NY, USA},
  doi = {10.1145/800157.805047},
  url = {https://dl.acm.org/doi/10.1145/800157.805047},
  urldate = {2023-09-06},
  isbn = {978-1-4503-7464-4}
}

@article{cubittComplexityClassificationLocal2016,
  title = {Complexity {{Classification}} of {{Local Hamiltonian Problems}}},
  author = {Cubitt, Toby and Montanaro, Ashley},
  year = 2016,
  journal = {SIAM Journal on Computing},
  volume = {45},
  number = {2},
  pages = {268--316},
  publisher = {{Society for Industrial and Applied Mathematics}},
  issn = {0097-5397},
  doi = {10.1137/140998287},
  url = {https://epubs.siam.org/doi/10.1137/140998287},
  urldate = {2023-09-06}
}

@inproceedings{falorCollapsiblePolynomialHierarchy2023,
  title = {A {{Collapsible Polynomial Hierarchy}} for {{Promise Problems}}},
  author = {Falor, Chirag and Ge, Shu and Natarajan, Anand},
  year = 2023,
  eprint = {2311.12228},
  primaryclass = {quant-ph},
  publisher = {arXiv:2311.12228},
  doi = {10.48550/arXiv.2311.12228},
  url = {http://arxiv.org/abs/2311.12228},
  urldate = {2024-08-28},
  archiveprefix = {arXiv}
}

@article{freedmanQuantumSystemsNonkhyperfinite2014,
  title = {Quantum Systems on Non-k-Hyperfinite Complexes: A Generalization of Classical Statistical Mechanics on Expander Graphs},
  shorttitle = {Quantum Systems on Non-k-Hyperfinite Complexes},
  author = {Freedman, Michael H. and Hastings, Matthew B.},
  year = 2014,
  journal = {Quantum Information \& Computation},
  volume = {14},
  number = {1-2},
  pages = {144--180},
  issn = {1533-7146}
}

@article{gharibianComplexitySimulatingLocal2019,
  title = {The Complexity of Simulating Local Measurements on Quantum Systems},
  author = {Gharibian, Sevag and Yirka, Justin},
  year = 2019,
  journal = {Quantum},
  volume = {3},
  pages = {189},
  publisher = {Verein zur F\"orderung des Open Access Publizierens in den Quantenwissenschaften},
  doi = {10.22331/q-2019-09-30-189},
  url = {https://doi.org/10.22331/q-2019-09-30-189},
  urldate = {2023-10-27},
  langid = {british}
}

@inproceedings{gharibianGroundStateConnectivity2015,
  title = {Ground {{State Connectivity}} of {{Local Hamiltonians}}},
  booktitle = {Automata, {{Languages}}, and {{Programming}} ({{ICALP}})},
  author = {Gharibian, Sevag and Sikora, Jamie},
  editor = {Halld{\'o}rsson, Magn{\'u}s M. and Iwama, Kazuo and Kobayashi, Naoki and Speckmann, Bettina},
  year = 2015,
  pages = {617--628},
  publisher = {Springer},
  address = {Berlin, Heidelberg},
  doi = {10.1007/978-3-662-47672-7_50},
  isbn = {978-3-662-47672-7},
  langid = {english}
}

@article{gharibianGuestColumn72024,
  title = {Guest {{Column}}: {{The}} 7 Faces of Quantum {{NP}}},
  shorttitle = {Guest {{Column}}},
  author = {Gharibian, Sevag},
  year = 2024,
  journal = {ACM SIGACT News},
  volume = {54},
  number = {4},
  pages = {54--91},
  issn = {0163-5700},
  doi = {10.1145/3639528.3639535},
  url = {https://doi.org/10.1145/3639528.3639535},
  urldate = {2024-01-07}
}

@inproceedings{gharibianHardnessApproximationQuantum2012,
  title = {Hardness of {{Approximation}} for {{Quantum Problems}}},
  booktitle = {Automata, {{Languages}}, and {{Programming}}},
  author = {Gharibian, Sevag and Kempe, Julia},
  editor = {Czumaj, Artur and Mehlhorn, Kurt and Pitts, Andrew and Wattenhofer, Roger},
  year = 2012,
  series = {Lecture {{Notes}} in {{Computer Science}}},
  pages = {387--398},
  publisher = {Springer},
  address = {Berlin, Heidelberg},
  doi = {10.1007/978-3-642-31594-7_33},
  url = {https://doi.org/10.1007/978-3-642-31594-7\_33},
  isbn = {978-3-642-31594-7},
  langid = {english}
}

@inproceedings{gharibianQuantumGeneralizationsPolynomial2018,
  title = {Quantum {{Generalizations}} of the {{Polynomial Hierarchy}} with {{Applications}} to {{QMA}}(2)},
  booktitle = {43rd {{International Symposium}} on {{Mathematical Foundations}} of {{Computer Science}} ({{MFCS}} 2018)},
  author = {Gharibian, Sevag and Santha, Miklos and Sikora, Jamie and Sundaram, Aarthi and Yirka, Justin},
  editor = {Potapov, Igor and Spirakis, Paul and Worrell, James},
  year = 2018,
  series = {Leibniz {{International Proceedings}} in {{Informatics}} ({{LIPIcs}})},
  volume = {117},
  pages = {58:1--58:16},
  publisher = {Schloss Dagstuhl--Leibniz-Zentrum fuer Informatik},
  address = {Dagstuhl, Germany},
  issn = {1868-8969},
  doi = {10.4230/LIPIcs.MFCS.2018.58},
  url = {http://drops.dagstuhl.de/opus/volltexte/2018/9640},
  urldate = {2023-10-27},
  isbn = {978-3-95977-086-6}
}

@inproceedings{gharibianQuantumSpaceGround2023,
  title = {Quantum {{Space}}, {{Ground Space Traversal}}, and {{How}} to {{Embed Multi-Prover Interactive Proofs}} into {{Unentanglement}}},
  booktitle = {14th {{Innovations}} in {{Theoretical Computer Science Conference}} ({{ITCS}} 2023)},
  author = {Gharibian, Sevag and Rudolph, Dorian},
  editor = {Tauman Kalai, Yael},
  year = 2023,
  series = {Leibniz {{International Proceedings}} in {{Informatics}} ({{LIPIcs}})},
  volume = {251},
  pages = {53:1--53:23},
  publisher = {Schloss Dagstuhl -- Leibniz-Zentrum f\"ur Informatik},
  address = {Dagstuhl, Germany},
  issn = {1868-8969},
  doi = {10.4230/LIPIcs.ITCS.2023.53},
  url = {https://drops.dagstuhl.de/opus/volltexte/2023/17556},
  urldate = {2023-06-20},
  isbn = {978-3-95977-263-1}
}

@inproceedings{gheorghiuEstimatingEntropyShallow2024,
  title = {On Estimating the Entropy of Shallow Circuit Outputs},
  author = {Gheorghiu, Alexandru and Hoban, Matty J.},
  year = 2024,
  eprint = {2002.12814},
  primaryclass = {quant-ph},
  publisher = {arXiv.2002.12814},
  doi = {10.48550/arXiv.2002.12814},
  url = {http://arxiv.org/abs/2002.12814},
  urldate = {2024-08-29},
  archiveprefix = {arXiv}
}

@article{gopalanConnectivityBooleanSatisfiability2009,
  title = {The {{Connectivity}} of {{Boolean Satisfiability}}: {{Computational}} and {{Structural Dichotomies}}},
  shorttitle = {The {{Connectivity}} of {{Boolean Satisfiability}}},
  author = {Gopalan, Parikshit and Kolaitis, Phokion G. and Maneva, Elitza and Papadimitriou, Christos H.},
  year = 2009,
  journal = {SIAM Journal on Computing},
  volume = {38},
  number = {6},
  pages = {2330--2355},
  publisher = {{Society for Industrial and Applied Mathematics}},
  issn = {0097-5397},
  doi = {10.1137/07070440X},
  url = {https://epubs.siam.org/doi/abs/10.1137/07070440X},
  urldate = {2024-08-29}
}

@article{gossetQCMAHardnessGround2017,
  title = {{{QCMA}} Hardness of Ground Space Connectivity for Commuting {{Hamiltonians}}},
  author = {Gosset, David and Mehta, Jenish C. and Vidick, Thomas},
  year = 2017,
  journal = {Quantum},
  volume = {1},
  pages = {16},
  publisher = {Verein zur F\"orderung des Open Access Publizierens in den Quantenwissenschaften},
  issn = {2521-327X},
  doi = {10.22331/q-2017-07-14-16},
  url = {https://doi.org/10.22331/q-2017-07-14-16}
}

@inproceedings{gossetQuantum3SATQMA1Complete2013,
  title = {Quantum 3-{{SAT Is QMA1-Complete}}},
  booktitle = {Proceedings of the 2013 {{IEEE}} 54th {{Annual Symposium}} on {{Foundations}} of {{Computer Science}}},
  author = {Gosset, David and Nagaj, Daniel},
  year = 2013,
  series = {{{FOCS}} '13},
  pages = {756--765},
  publisher = {IEEE Computer Society},
  address = {USA},
  doi = {10.1109/FOCS.2013.86},
  url = {https://doi.org/10.1109/FOCS.2013.86},
  urldate = {2023-10-27},
  isbn = {978-0-7695-5135-7}
}

@inproceedings{gottesmanQuantumClassicalComplexity2009,
  title = {The {{Quantum}} and {{Classical Complexity}} of {{Translationally Invariant Tiling}} and {{Hamiltonian Problems}}},
  booktitle = {2009 50th {{Annual IEEE Symposium}} on {{Foundations}} of {{Computer Science}}},
  author = {Gottesman, Daniel and Irani, Sandy},
  year = 2009,
  pages = {95--104},
  issn = {0272-5428},
  doi = {10.1109/FOCS.2009.22},
  url = {https://ieeexplore.ieee.org/document/5438643},
  urldate = {2023-10-11}
}

@inproceedings{grewalEntangledQuantumPolynomial2024,
  title = {The {{Entangled Quantum Polynomial Hierarchy Collapses}}},
  booktitle = {39th {{Computational Complexity Conference}} ({{CCC}} 2024)},
  author = {Grewal, Sabee and Yirka, Justin},
  year = 2024,
  publisher = {Schloss Dagstuhl -- Leibniz-Zentrum f\"ur Informatik},
  doi = {10.4230/LIPIcs.CCC.2024.6},
  url = {https://drops.dagstuhl.de/entities/document/10.4230/LIPIcs.CCC.2024.6},
  urldate = {2024-08-28},
  copyright = {https://creativecommons.org/licenses/by/4.0/legalcode},
  langid = {english}
}

@article{hallgrenLocalHamiltonianProblem2013,
  title = {The Local {{Hamiltonian}} Problem on a Line with Eight States Is {{QMA-complete}}},
  author = {Hallgren, Sean and Nagaj, Daniel and Narayanaswami, Sandeep},
  year = 2013,
  journal = {Quantum Information \& Computation},
  volume = {13},
  number = {9-10},
  pages = {721--750},
  issn = {1533-7146}
}

@inproceedings{hiraharaProbabilisticallyCheckableReconfiguration2024,
  title = {Probabilistically {{Checkable Reconfiguration Proofs}} and {{Inapproximability}} of {{Reconfiguration Problems}}},
  booktitle = {Proceedings of the 56th {{Annual ACM Symposium}} on {{Theory}} of {{Computing}}},
  author = {Hirahara, Shuichi and Ohsaka, Naoto},
  year = 2024,
  series = {{{STOC}} 2024},
  pages = {1435--1445},
  publisher = {Association for Computing Machinery},
  address = {New York, NY, USA},
  doi = {10.1145/3618260.3649667},
  url = {https://dl.acm.org/doi/10.1145/3618260.3649667},
  urldate = {2024-09-05},
  isbn = {979-8-4007-0383-6}
}

@inproceedings{iraniCommutingLocalHamiltonian2023,
  title = {Commuting {{Local Hamiltonian Problem}} on {{2D}} beyond Qubits},
  author = {Irani, Sandy and Jiang, Jiaqing},
  year = 2023,
  eprint = {2309.04910},
  primaryclass = {quant-ph},
  publisher = {arXiv:2309.04910},
  doi = {10.48550/arXiv.2309.04910},
  url = {http://arxiv.org/abs/2309.04910},
  urldate = {2023-10-12},
  archiveprefix = {arXiv}
}

@article{itoComplexityReconfigurationProblems2011,
  title = {On the Complexity of Reconfiguration Problems},
  author = {Ito, Takehiro and Demaine, Erik D. and Harvey, Nicholas J. A. and Papadimitriou, Christos H. and Sideri, Martha and Uehara, Ryuhei and Uno, Yushi},
  year = 2011,
  journal = {Theoretical Computer Science},
  volume = {412},
  number = {12},
  pages = {1054--1065},
  issn = {0304-3975},
  doi = {10.1016/j.tcs.2010.12.005},
  url = {https://www.sciencedirect.com/science/article/pii/S0304397510006961},
  urldate = {2024-09-05}
}

@article{kempeComplexityLocalHamiltonian2006,
  title = {The {{Complexity}} of the {{Local Hamiltonian Problem}}},
  author = {Kempe, Julia and Kitaev, Alexei and Regev, Oded},
  year = 2006,
  journal = {SIAM Journal on Computing},
  volume = {35},
  number = {5},
  pages = {1070--1097},
  publisher = {{Society for Industrial and Applied Mathematics}},
  issn = {0097-5397},
  doi = {10.1137/S0097539704445226},
  url = {https://epubs.siam.org/doi/10.1137/S0097539704445226},
  urldate = {2023-10-10}
}

@article{kempejulia3localHamiltonianQMAcomplete2003,
  title = {3-Local {{Hamiltonian}} Is {{QMA-complete}}},
  author = {{Kempe, Julia} and {Regev, Oded}},
  year = 2003,
  journal = {Quantum Information \& Computation},
  volume = {3},
  number = {3},
  pages = {258--264}
}

@book{kitaevClassicalQuantumComputation2002,
  title = {Classical and {{Quantum Computation}}},
  author = {Kitaev, A. and Shen, A. and Vyalyi, M.},
  year = 2002,
  series = {Graduate {{Studies}} in {{Mathematics}}},
  volume = {47},
  publisher = {American Mathematical Society},
  issn = {1065-7339},
  doi = {10.1090/gsm/047},
  url = {https://www.ams.org/gsm/047},
  urldate = {2023-09-21},
  isbn = {978-0-8218-3229-5 978-1-4704-0927-2 978-1-4704-2011-6 978-1-4704-1800-7},
  langid = {english}
}

@article{kitaevFaulttolerantQuantumComputation2003,
  title = {Fault-Tolerant Quantum Computation by Anyons},
  author = {Kitaev, A. {\relax Yu}.},
  year = 2003,
  journal = {Annals of Physics},
  volume = {303},
  number = {1},
  pages = {2--30},
  issn = {0003-4916},
  doi = {10.1016/S0003-4916(02)00018-0},
  url = {https://www.sciencedirect.com/science/article/pii/S0003491602000180},
  urldate = {2023-09-21}
}

@article{kitaevUnpairedMajoranaFermions2001,
  title = {Unpaired {{Majorana}} Fermions in Quantum Wires},
  author = {Kitaev, A. Yu},
  year = 2001,
  journal = {Physics-Uspekhi},
  volume = {44},
  number = {10S},
  pages = {131},
  issn = {1063-7869},
  doi = {10.1070/1063-7869/44/10S/S29},
  url = {https://dx.doi.org/10.1070/1063-7869/44/10S/S29},
  urldate = {2023-09-21},
  langid = {english}
}

@article{leonidlevinUniversalSearchProblems1973,
  title = {Universal Search Problems},
  author = {{Leonid Levin}},
  year = 1973,
  journal = {Problems of Information Transmission},
  volume = {9},
  number = {3},
  pages = {265--266}
}

@inproceedings{lockhartQuantumStateIsomorphism2017,
  title = {Quantum {{State Isomorphism}}},
  author = {Lockhart, Joshua and Guill{\'e}n, Carlos E. Gonz{\'a}lez},
  year = 2017,
  eprint = {1709.09622},
  primaryclass = {quant-ph},
  publisher = {arXiv:1709.09622},
  doi = {10.48550/arXiv.1709.09622},
  url = {http://arxiv.org/abs/1709.09622},
  urldate = {2024-08-28},
  archiveprefix = {arXiv}
}

@article{natarajanStatusQuantumPCP2024,
  title = {The Status of the Quantum {{PCP}} Conjecture (Games Version)},
  author = {Natarajan, Anand and Nirkhe, Chinmay},
  year = 2024,
  number = {arXiv:2403.13084},
  eprint = {2403.13084},
  primaryclass = {quant-ph},
  publisher = {arXiv},
  doi = {10.48550/arXiv.2403.13084},
  url = {http://arxiv.org/abs/2403.13084},
  urldate = {2025-05-01},
  archiveprefix = {arXiv},
  journal = {10.48550/arXiv.2403.13084}
}

@article{nishimuraIntroductionReconfiguration2018,
  title = {Introduction to {{Reconfiguration}}},
  author = {Nishimura, Naomi},
  year = 2018,
  journal = {Algorithms},
  volume = {11},
  number = {4},
  pages = {52},
  doi = {10.3390/a11040052},
  url = {https://doi.org/10.3390/a11040052}
}

@inproceedings{ohsakaGapPreservingReductions2023,
  title = {Gap {{Preserving Reductions Between Reconfiguration Problems}}},
  booktitle = {40th {{International Symposium}} on {{Theoretical Aspects}} of {{Computer Science}} ({{STACS}} 2023)},
  author = {Ohsaka, Naoto},
  year = 2023,
  publisher = {Schloss Dagstuhl -- Leibniz-Zentrum f\"ur Informatik},
  doi = {10.4230/LIPIcs.STACS.2023.49},
  url = {https://drops.dagstuhl.de/entities/document/10.4230/LIPIcs.STACS.2023.49},
  urldate = {2024-09-05},
  copyright = {https://creativecommons.org/licenses/by/4.0/legalcode},
  langid = {english}
}

@article{regevLatticesLearningErrors2009,
  title = {On Lattices, Learning with Errors, Random Linear Codes, and Cryptography},
  author = {Regev, Oded},
  year = 2009,
  journal = {J. ACM},
  volume = {56},
  number = {6},
  pages = {34:1--34:40},
  issn = {0004-5411},
  doi = {10.1145/1568318.1568324},
  url = {https://doi.org/10.1145/1568318.1568324},
  urldate = {2024-08-29}
}

@inproceedings{s.InapproximabilityReconfigurationProblems2024,
  title = {On {{Inapproximability}} of {{Reconfiguration Problems}}: {{PSPACE-Hardness}} and Some {{Tight NP-Hardness Results}}},
  shorttitle = {On {{Inapproximability}} of {{Reconfiguration Problems}}},
  author = {S., Karthik C. and Manurangsi, Pasin},
  year = 2024,
  eprint = {2312.17140},
  primaryclass = {cs},
  publisher = {arXiv:2312.17140},
  doi = {10.48550/arXiv.2312.17140},
  url = {http://arxiv.org/abs/2312.17140},
  urldate = {2024-09-06},
  archiveprefix = {arXiv}
}

@inproceedings{schaeferComplexitySatisfiabilityProblems1978,
  title = {The Complexity of Satisfiability Problems},
  booktitle = {Proceedings of the Tenth Annual {{ACM}} Symposium on {{Theory}} of Computing},
  author = {Schaefer, Thomas J.},
  year = 1978,
  series = {{{STOC}} '78},
  pages = {216--226},
  publisher = {Association for Computing Machinery},
  address = {New York, NY, USA},
  doi = {10.1145/800133.804350},
  url = {https://dl.acm.org/doi/10.1145/800133.804350},
  urldate = {2023-10-10},
  isbn = {978-1-4503-7437-8}
}

@article{schuchComplexityCommutingHamiltonians2011,
  title = {Complexity of Commuting {{Hamiltonians}} on a Square Lattice of Qubits},
  author = {Schuch, Norbert},
  year = 2011,
  journal = {Quantum Information \& Computation},
  volume = {11},
  number = {11-12},
  pages = {901--912},
  issn = {1533-7146}
}

@inproceedings{umansHardnessApproximatingSpl1999,
  title = {Hardness of Approximating /Spl {{Sigma}}//Sub 2//Sup p/ Minimization Problems},
  booktitle = {40th {{Annual Symposium}} on {{Foundations}} of {{Computer Science}} ({{Cat}}. {{No}}.{{99CB37039}})},
  author = {Umans, C.},
  year = 1999,
  pages = {465--474},
  issn = {0272-5428},
  doi = {10.1109/SFFCS.1999.814619},
  url = {https://ieeexplore.ieee.org/document/814619},
  urldate = {2023-10-13}
}

@book{vaziraniApproximationAlgorithms2003,
  title = {Approximation {{Algorithms}}},
  author = {Vazirani, Vijay V.},
  year = 2003,
  publisher = {Springer},
  address = {Berlin, Heidelberg},
  doi = {10.1007/978-3-662-04565-7},
  url = {http://link.springer.com/10.1007/978-3-662-04565-7},
  urldate = {2026-04-24},
  copyright = {http://www.springer.com/tdm},
  isbn = {978-3-642-08469-0 978-3-662-04565-7},
  langid = {english}
}

@inproceedings{watrousLimitsPowerQuantum2002,
  title = {Limits on the Power of Quantum Statistical Zero-Knowledge},
  booktitle = {The 43rd {{Annual IEEE Symposium}} on {{Foundations}} of {{Computer Science}}, 2002. {{Proceedings}}.},
  author = {Watrous, J.},
  year = 2002,
  pages = {459--468},
  issn = {0272-5428},
  doi = {10.1109/SFCS.2002.1181970},
  url = {https://ieeexplore.ieee.org/document/1181970},
  urldate = {2024-08-29}
}

@inproceedings{watsonComplexityTranslationallyInvariant2023,
  title = {The {{Complexity}} of {{Translationally Invariant Problems Beyond Ground State Energies}}},
  booktitle = {40th {{International Symposium}} on {{Theoretical Aspects}} of {{Computer Science}} ({{STACS}} 2023)},
  author = {Watson, James D. and Bausch, Johannes and Gharibian, Sevag},
  editor = {Berenbrink, Petra and Bouyer, Patricia and Dawar, Anuj and Kant{\'e}, Mamadou Moustapha},
  year = 2023,
  series = {Leibniz {{International Proceedings}} in {{Informatics}} ({{LIPIcs}})},
  volume = {254},
  pages = {54:1--54:21},
  publisher = {Schloss Dagstuhl -- Leibniz-Zentrum f\"ur Informatik},
  address = {Dagstuhl, Germany},
  issn = {1868-8969},
  doi = {10.4230/LIPIcs.STACS.2023.54},
  url = {https://drops.dagstuhl.de/opus/volltexte/2023/17706},
  urldate = {2023-10-26},
  isbn = {978-3-95977-266-2}
}

@inproceedings{wocjanSeveralNaturalBQPComplete2006,
  title = {Several Natural {{BQP-Complete}} Problems},
  author = {Wocjan, Pawel and Zhang, Shengyu},
  year = 2006,
  eprint = {quant-ph/0606179},
  publisher = {arXiv:quant-ph/0606179},
  doi = {10.48550/arXiv.quant-ph/0606179},
  url = {http://arxiv.org/abs/quant-ph/0606179},
  urldate = {2023-10-12},
  archiveprefix = {arXiv}
}

@incollection{yamakamiQuantumNPQuantum2002,
  title = {Quantum {{NP}} and a {{Quantum Hierarchy}}},
  booktitle = {Foundations of {{Information Technology}} in the {{Era}} of {{Network}} and {{Mobile Computing}}},
  author = {Yamakami, Tomoyuki},
  editor = {{Baeza-Yates}, Ricardo and Montanari, Ugo and Santoro, Nicola},
  year = 2002,
  series = {{{IFIP}} --- {{The International Federation}} for {{Information Processing}}},
  pages = {323--336},
  publisher = {Springer US},
  address = {Boston, MA},
  doi = {10.1007/978-0-387-35608-2_27},
  url = {https://doi.org/10.1007/978-0-387-35608-2\_27},
  urldate = {2023-09-21},
  isbn = {978-0-387-35608-2},
  langid = {english}
}

@preamble{ "\providecommand{\noopsort}[1]{} " }

@book{Kit,
author = {Kitaev, A. Yu. and Shen, A. H. and Vyalyi, M. N.},
title = {Classical and Quantum Computation},
year = {2002},
isbn = {0821832298},
publisher = {American Mathematical Society},
address = {USA}
}

@article{Fannes:1973ddo,
    author = "Fannes, M.",
    title = "{A continuity property of the entropy density for spin lattice systems}",
    doi = "10.1007/BF01646490",
    journal = "Commun. Math. Phys.",
    volume = "31",
    number = "4",
    pages = "291--294",
    year = "1973"
}

\end{document}